\documentclass{article}
\usepackage{amsfonts}
\usepackage{amsmath}
\usepackage{xcolor}
\newtheorem{theorem}{Theorem}

\newtheorem{assumption}{Assumption}

\newtheorem{corollary}[theorem]{Corollary}

\newtheorem{definition}[theorem]{Definition}

\newtheorem{lemma}[theorem]{Lemma}

\newtheorem{proposition}[theorem]{Proposition}
\newtheorem{remark}[theorem]{Remark}

\newenvironment{proof}[1][Proof]{\noindent \textbf{#1.} }{\  \rule{0.5em}{0.5em}}
\DeclareMathOperator*{\esssup}{ess\;sup}

\newcommand{\bq}{\begin{equation*}}
\newcommand{\eq}{\end{equation*}}
\newcommand{\bqn}{\begin{equation}}
\newcommand{\eqn}{\end{equation}}
\newcommand{\bqq}{\begin{eqnarray*}}
\newcommand{\eqq}{\end{eqnarray*}}
\newcommand{\bqqn}{\begin{eqnarray}}
\newcommand{\eqqn}{\end{eqnarray}}
\DeclareMathOperator*{\essinf}{ess\;inf}

\begin{document}

\title{A game theoretical approach to homothetic robust forward investment performance processes  in
stochastic factor models
}
\author{Juan Li\thanks{%
School of Mathematics and Statistics, Shandong University, Weihai, Weihai, China. Partially supported by National Key R and D Program of China (No.2018YFA0703900), by the NSF of P.R. China (No. 11871037) and by NSFC-RS (No. 11661130148; NA150344). Email: \texttt{juanli@sdu.edu.cn}.
} \and Wenqiang Li\thanks{%
 School of Mathematics and Information Sciences, Yantai University, Yantai, China. Partially supported by Natural Science Foundation of Shandong Province (No.ZR2017MA015) and by Doctoral Scientific Research Fund of Yantai University (No. SX17B09). Email: \texttt{wenqiangli@ytu.edu.cn}. Corresponding author.
} \and
Gechun Liang\thanks{%
Department of Statistics, The University of Warwick, Coventry, U.K.  Partially supported by NNSF of China (Grant No. 11771158, 11801091). Email: \texttt{g.liang@warwick.ac.uk}.}}
\date\today
\maketitle

\begin{abstract}
This paper studies an optimal forward investment problem in an
incomplete market with model uncertainty, in which the underlying
stocks depend on the correlated stochastic factors. The uncertainty
stems from  the  probability measure chosen by an investor to
evaluate the performance. We
obtain directly the representation of the homothetic robust forward performance processes in factor-form by
combining the zero-sum stochastic differential game and ergodic BSDE approach.
We also establish the connections with the risk-sensitive zero-sum stochastic differential games over an infinite horizon
with ergodic payoff criteria, as well as with the classical robust expected utilities for long time horizons. Finally, we give an example to illustrate that  our approach can be applied to address a type of robust forward investment  performance processes with negative realization processes.  \\


\end{abstract}

\section{{Introduction}}

The aim of this paper is to study optimal investment evaluated by a
forward performance criterion in a stochastic factor market model,
in which the probability measure that models future stock price
evolutions is ambiguous. The forward performance process, as an
adapted stochastic dynamic utility evolving forward in time, has
been introduced and developed in \cite{MZ0}-\cite{MZ3} (see also
\cite{Henderson-Hobson} and \cite{Gordan}, and more recently
\cite{bahman}, \cite{Michalis}, \cite{BCZ2018}, \cite{CHLZ},
\cite{HeStrub}, \cite{HLT}, \cite{Sigrid}, \cite{LSZ},
\cite{LSZ2020} and \cite{SZ}). This new concept differs from the
classical expected utility function, in which the objective is to
solve a stochastic control problem in a backward way via dynamic
programming principle. One of the advantages of forward performance
processes is allowing the investor to consider optimal investment
problems with arbitrary horizons. As such, it provides a useful
complement and a natural extension to the classical expected utility
function.

Recall the classical expected utility theory for the
optimal portfolio selection is
$$\sup_{\pi}E_{\mathbb{P}}[U(X_T^\pi)],$$
where $\pi$ is the portfolio choice, $\mathbb{P}$ is a probability
measure that is used to measure the evolutions of stock prices, $T$
is the terminal horizon, and $U$ is a \emph{fixed} utility function
at time $T$. In spite of the popularity of expected utility theory,
there has been some criticism of it. One of them is the fact that it
is not satisfactory in dealing with model uncertainty (also called
Knightian uncertainty) as predicted by the famous Ellsberg paradox.
In fact, an investor frequently faces significant ambiguity about
the probability measure $\mathbb{P}$ to evaluate the investment
performance. In finance, \cite{Lo} argued that the (perceived)
failures of the dominant paradigm, for example, in the context of
the recent crisis, are due to inadequate attention paid to the kind
of uncertainty faced by agents and modelers.

One possible way to address this problem is to use the concept of
robust utility, which was introduced to account for uncertain
aversion. It can be numerically represented by the following form
$$X\rightarrow\inf_{\mathbb{P}\in\mathcal{P}}E_{\mathbb{P}}[U(X)],$$
where $\mathcal{P}$ is a family of probability measures describing
all the possible probabilities of future scenarios and the infimum
means the worst-case scenario is implemented. Robust utility
maximization in the optimal investment problems has been widely
investigated under different situations with different approaches,
among others, a stochastic control method  in \cite{BMS, HS2006}, a
stochastic differential game approach in \cite{OS2011}, a duality
method in \cite{S2007}.  For more details on various portfolio
selection problems,  we refer to \cite{AM2019, DCH2018, FG, FSW,
HS2007, YLZ, YYY2016} and the references therein. In particular, we
refer to \cite{KKF2014, KKF2017, KKF2018}  for the review of the
recent advancements in robust investment management.


In this paper, we consider the ambiguity of the probability measure
under the framework of forward performance processes in incomplete
markets. We propose a framework that solves directly the above
problem in a unified manner, combining the \emph{zero-sum stochastic
differential game} and \emph{ergodic backward stochastic
differential equation} (BSDE) theory. The concept of robust forward
performance processes was recently introduced in \cite{Jan}, by
using a penalty function to weight relatively the probability models
such that they are more in line with the actual market. They
obtained the characterization of the robust forward criteria via a
duality approach. See also \cite{CL} for an extension to uncertain
parameters. However, both papers only consider robust forward
performance processes with zero volatility, in particular, the
\emph{Markovian} case for the stochastic factor model is not
covered. In this paper, we consider the \emph{Markovian} robust
forward performance process in a stochastic factor model. The
approach is different from the duality approach used in \cite{Jan}
and the saddle point method used in \cite{CL}. Next, we briefly
introduce our framework and explain our major contributions.

We construct the robust forward performance process via a two-player
zero-sum stochastic differential game. In our model, the ambiguity
of the probability measure is described via a family of equivalent
probability measures parameterized by a density process $u$ in a
compact and convex set (see \eqref{112901}). We parameterize the robust forward performance process
by the density process $u$. This generalizes the original definition of the robust forward performance process introduced in
\cite{Jan} with penalty functions as a special class of parametrization. We refer to Definition \ref{def} and Remarks \ref{re-def1}-\ref{re-def2} for more details.

To robustify the optimal investment, the investor will select the
best investment portfolio that is least affected by the model
uncertainty,
 whereas the nature of the market  acts to minimize the
expected forward preference by choosing the worst-case scenario.
This leads to a two-player zero-sum stochastic differential game
{between the investor and the market},
where each player's decision (strategy) depends on the
counterparty's action (control) she has observed. Therefore, the concept
of ``strategy" corresponding to the ``control" will play a key role
in analyzing the game (see \cite{BL, FS1989}).

Utilizing the idea of ``strategy", we give a new characterization of
the robust forward performance process. Specifically, both the
worst-case scenario ``strategy" corresponding to each portfolio
selection and the optimal investment policy under the worst-case
scenario ``strategy" are given in our characterization (see
(\ref{beta_star})-(\ref{martingale_maxmin})). Moreover, if the game
value exists, the optimal investment ``strategy" corresponding to
each scenario and the worst-case scenario under the optimal
investment ``strategy" are further given  in
(\ref{alpha_star})-(\ref{martingale_minmax}). Compared to the saddle
point argument used in \cite{HS2006, HS2007, YLZ} in the classical
framework and \cite{CL} in the forward framework, our
characterization (\ref{beta_star})-\eqref{martingale_maxmin} and \eqref{alpha_star}-(\ref{martingale_minmax}) relies
on the investor's response to each scenario and portfolio choice.
Moreover, it is often relatively easy to compute the optimal
strategies, as they only involve maximization/minimization problems
rather than maxmin/minmax problems. On the other hand, our
stochastic differential game approach may also provide an
alternative way to study the robust forward performance process when
saddle point does not exist. We present an example for the case
$\tau<0$ (i.e. with a {negative realization process}) in Section 7 to illustrate
this point.

The second component to construct the robust forward performance
process in factor form is an ergodic BSDE. The stochastic PDE (SPDE)
approach, introduced in \cite{MZ3} to characterize the forward
performance processes (without model ambiguity), may not be applied
directly to our model. First, the form of the related SPDE is not
easy to derive due to the presence of model uncertainty. Second, it
is difficult to obtain the solution existence and uniqueness of the
SPDE for the general case even if we know the form of the equation.

In order to get the representation of the {homothetic} robust forward performance
process in stochastic factor form, we apply directly the ergodic
BSDE approach, which was first proposed in \cite{HU2} to study
ergodic control problems. The ergodic BSDE approach was first
exploited in \cite{LZ} to study the representation of the homothetic
forward performance process in the absence of model uncertainty. We
first characterize the power robust forward performance process in
terms of the solution of an ill-posed Isaacs type equation. Although
the solution of this Isaacs equation can not be obtained directly,
it offers (i) the construction of the optimal portfolio ``strategy",
the worst-case scenario ``strategy", and the related optimal
portfolio choice and the worst-case scenario; (ii) the hint of the
driver form of the corresponding ergodic BSDE. Then, we obtain the
representation of the robust power forward performance process by
using the Markovian solution of the ergodic BSDE. The associated
optimal portfolio and worst-case scenario ``strategy" and ``control"
are also obtained in feedback form of the stochastic factor. {We can also obtain
other type of homothetic robust forward performance processes (logarithmic and exponential) using the same approach.}

The third contribution of this paper is establishing a connection
between the constant $\lambda$ appearing in the solution of the
ergodic BSDE (\ref{EQBSDE1}) and a class of zero-sum risk-sensitive
stochastic differential game over an infinite horizon with ergodic
payoff criteria.  Risk-sensitive optimal control has been widely
applied to optimal investment problems (see, \cite{Bielecki, FS2,
FS3, Henderson3} and references therein). The corresponding
risk-sensitive stochastic differential games are studied in
\cite{B1999, BG2012, BS2018, Knispel} via PDE approach and in
\cite{KH2003} via BSDE approach.

In this paper, we apply directly the ergodic BSDE approach to
address the zero-sum risk-sensitive stochastic differential game
with ergodic payoff criteria over an infinite horizon. Thus, we
provide a new method to obtain the value of the risk-sensitive game
problem and give the robust optimal investment policy which
generalizes the results in \cite{FS2, FS3} to the stochastic factor
model with uncertainty. To obtain this connection, we prove a
comparison result for a class of ergodic BSDE whose drivers are only
locally Lipschitz continuous. With the help of this connection, the
constant $\lambda$ can be interpreted as the \emph{optimal long-term
growth rate of the expected utility with model uncertainty}, and can
also be applied to study the related ``robust large deviations"
criteria for long-term investment problems.

In addition, we develop a connection between the robust forward
performance process and classical robust expected utility. Optimal
investment problems with classical robust expected utilities have
been studied via different methods, among others, by the duality
approach \cite{FG, S2007},  the stochastic control approach based on
BSDE \cite{BMS, HS2006} and stochastic differential game approach
based on PDE \cite{TZ2002}. With the help of the relation
established in \cite{Hu11} on the solution of finite horizon BSDE
and the solution of associated ergodic BSDE, we prove that an
appropriately discounted lower value function associated with the
classical power robust expected utility will converge to the power
robust forward performance process as the trading horizon tends to
infinity.

This paper is organized as follows. In section \ref{SectionModel},
we introduce the market model with uncertainty and  the notion of
robust forward performance processes.
{The stochastic differential game approach is given in subsections \ref{subsection2.1} and \ref{subsection2.2} for different situations.} In section
\ref{SectionPowerUtility}, we focus on the power case and construct
the robust forward performance process in factor-form. Two examples
are  given in section \ref{examples} to illustrate the applications
in incomplete markets. Then,  we present the connections with the
risk-sensitive game problem and classical expected utility in
sections \ref{Connection with risk-sensitive optimization} and
\ref{SubSectionTraditionalPowerUtility}, respectively. In section 7,
we further provide an example with a {negative realization process} for which
saddle point does not exist. Finally, section 8 concludes.




\section{The stochastic factor model with uncertainty and its robust forward performance process}

\label{SectionModel}

Let $(\Omega ,%
\mathcal{F},\mathbb{F}=\{ \mathcal{F}_{t}\}_{t\geq 0},\mathbb{P})$ be a filtered probability space satisfying the usual conditions, on which the process $W=(W^{1},\cdots ,W^{d})^{T}$ is a standard $%
d$-dimensional Brownian motion. Here, the superscript $T$ denotes the matrix transpose.
Suppose the market consists of a risk-free bond and $n$ risky stocks. The bond is assumed to be zero interest rate and the discounted (by the bond) individual  stock price $%
S_{t}^{i},$ $t\geq 0,$ affected by the stochastic factor process $V$, has the following form, for $i=1,...,n,$%
\begin{equation}
\frac{dS_{t}^{i}}{S_{t}^{i}}=b^{i}(V_{t})dt+\sum_{j=1}^{d}\sigma
^{ij}(V_{t})dW_{t}^{j},  \label{stock-SDE}
\end{equation}%
with $S_{0}^{i}>0,$ where the factor process $V=(V^{1},\cdots ,V^{d})^T$  satisfies, for $i=1,...,d,$%
\begin{equation}
dV_{t}^{i}=\eta ^{i}(V_{t})dt+\sum_{j=1}^{d}\kappa ^{ij}dW_{t}^{j},
\label{factor-SDE}
\end{equation}%
with $V_{0}^{i}\in \mathbb{R}$.  We introduce the basic  assumptions on the above model.

\begin{assumption}
\label{Assumption}

(H1) The coefficients
$$b: \mathbb{R}^d\rightarrow\mathbb{R}^n,\ \sigma: \mathbb{R}^d\rightarrow\mathbb{R}^{n\times d},$$  are
uniformly bounded and the volatility matrix $\sigma (v)$ has full row rank $%
n $.

(H2) The drift coefficient $\eta$ satisfies the following dissipative
condition: {there exists some positive constant $C_{\eta }>0$ such that}
\begin{equation}
(\eta (v)-\eta (\bar{v}))^{T}(v-\bar{v})\leq -C_{\eta }|v-\bar{v}|^{2},
\label{dissipative}
\end{equation}%
for any $v,\bar{v}\in \mathbb{R}^{d}$. The volatility matrix $\kappa =(\kappa ^{ij})$, $1\leq
i,j\leq d$, is a constant matrix with $\kappa \kappa ^{T}$ positive
definite and normalized to $||\kappa
||=1$\footnote{{Herein, $||\cdot||$
represents the trace norm for the matrix in $\mathbb{R}^{d\times
d}$, whereas $|\cdot|$ in \eqref{dissipative} represents the norm for
the vectors in $\mathbb{R}^d$.}}.
\end{assumption}
{Herein, the dissipative condition (\ref{dissipative}) is introduced to
ensure the existence  of a unique invariant measure of the
stochastic factor process $V$, i.e., $V$ is ergodic.
To simplify the notation, we introduce  the market price of risk vector $\theta (v),$ which is
defined as
\begin{equation}
\theta (v)=\sigma (v)^{T}[\sigma (v)\sigma (v)^{T}]^{-1}b(v),\ v\in \mathbb{R}^{d},
\end{equation}
so it solves the market price of risk equation $\sigma (v)\theta (v)=b(v)$. In addition, we suppose the market price of risk vector $\theta (v),$ $v\in \mathbb{R}^{d},$
 is
uniformly bounded and Lipschitz continuous with bound $K_{\theta}$ and Lipschitz constant $C_{\theta}$.}

We consider an investor  starting at time $t=0$ with initial
wealth level $x>0$ and trading among the bond and the stocks. Let
$\tilde{\pi}=(\tilde{\pi}^{1},\cdots ,\tilde{\pi}^{n})^{T}$ be the proportions
of her total wealth in the individual stock
accounts. Then, due to the self-financing policy,  the cumulative wealth
process $X^\pi$ satisfies
\begin{equation*}
dX_{t}^{\pi }=\sum \limits_{i=1}^{n}\frac{%
\tilde{\pi}_{t}^{i}X_{t}^{\pi }}{S_{t}^{i}}dS_{t}^{i}=X_{t}^{\pi }\tilde{\pi}_{t}^{T}\left(
b(V_{t})dt+\sigma (V_{t})dW_{t}\right) .
\end{equation*}%
As in \cite{LZ}, using the investment proportions rescaled by the volatility of stock prices, namely,
$\pi _{t}^{T}=\tilde{\pi}_{t}^{T}\sigma (V_{t}),$
we get
\begin{equation}
dX_{t}^{\pi }=X_{t}^{\pi }\pi _{t}^{T}(\theta (V_{t})dt+dW_{t}),
\label{wealth-process}
\end{equation}%
with $X_{0}^\pi=x\in \mathbb{R}_+$.

Next, we consider model uncertainty, i.e., the ambiguity of the probability measure which evaluates the performance. We denote by $\tilde{u}=(\tilde{u}^1,\cdots,\tilde{u}^d)^T$ the  parameters reflecting the possible future scenarios. For convenience, we will work throughout with the scenario parameters rescaled by the volatility of  stochastic factors, i.e., $$u=\kappa^T(\kappa\kappa^T)^{-1}\tilde{u}.$$
We introduce admissible spaces $\Tilde{\Pi}$ and $\mathcal{U}$ for the rescaled investment proportions $\pi$ and scenario parameters $u$, respectively.
\begin{definition}
Let $\Pi\subset\mathbb{R}^d$ be convex and closed and include the origin $0$. For any $t\geq 0$,  a process $\pi:\Omega\times [0,t]\rightarrow{\Pi}$ is an admissible investment proportion for an investor in the trading interval $[0,t]$, if $\pi\in\mathcal{L}_{BMO}^{2}[0,t]$, where
\begin{equation*}
\begin{aligned}
&\mathcal{L}_{BMO}^{2}[0,t]=\Big\{ (\pi _{s})_{s\in \lbrack 0,t]}:\pi \
\text{is}\  \mathbb{F}\text{-progressively\ measurable,}\ \\
& E_{\mathbb{P}}(  \int_{\tau }^{t}|\pi
_{s}|^{2}ds \vert \mathcal{F}_{\tau }) \leq C, a.s.,\  \text{for some constant C and all}%
\  \mathbb{F}\text{-stopping times}\ \tau\leq t\Big\}.
\end{aligned}
\end{equation*}%
The set of all admissible investment proportions in the trading interval $[0,t]$ is denoted by
 $\Pi_{{\left[ 0,t\right] }}$. Moreover, we define the set of \textit{admissible proportions} for all time horizons as $\tilde{\Pi}:=\cup _{t\geq 0}\Pi_{[0,t]}$.

\end{definition}
\begin{definition}
Suppose that $U\subseteq\mathbb{R}^d$ is convex and compact. For any $s\geq t\geq 0$, a process $u:[t,s]\times\Omega\rightarrow U$ is an  admissible scenario parameter if it is $\mathbb{F}$-progressively measurable and essentially bounded. We denote by $\mathcal{U}_{t,s}$ and $\mathcal{U}$ the set of all admissible scenario parameter in the time interval $[t,s]$ and for all time horizons, respectively.
\end{definition}

{In a market with model uncertainty, the investor will apply $\mathbb{P}^u$ to measure her preference instead of the probability measure $\mathbb{P}$, where the probability measure $\mathbb{P}^u$ is an equivalent probability measure with respect to $\mathbb{P}$ and  introduced by the following measure transformation
\begin{equation}\label{112901}
\frac{d\mathbb{P}^u}{d\mathbb{P}}\Big|_{\mathcal{F}_t}
=\mathcal{E}\left(\int_0^tu_s^TdW_s\right):=\exp\{\int_0^tu_s^TdW_s-\frac{1}{2}\int_0^t|u_s|^2ds\},\ u\in\mathcal{U}.
\end{equation}
In fact, this characterization of model uncertainty, admitting an
entire class $\{\mathbb{P}^u|u\in\mathcal{U}\}$ of possible prior
models, is a common approach applied in the  classical robust
expected utility, see  \cite{HS2006}.}

For every $u\in\mathcal{U}$, the process $W^u$ defined as
\begin{equation}\label{new-B.M.}
dW_t^u=-u_tdt+dW_t,
\end{equation} is a Brownian motion under the probability measure $\mathbb{P}^u$. Moreover, if $\pi\in\Pi_{[0,t]}$, then under $\mathbb{P}^u$, we also have
$$ess\sup_{\tau }E_{\mathbb{P}^u}\left( \left. \int_{\tau }^{t}|\pi
_{s}|^{2}ds\right \vert \mathcal{F}_{\tau }\right) <\infty.$$

The investor will evaluate her investment via a forward performance
process, the concept of which was first introduced and developed
in \cite{MZ0}-\cite{MZ3}. Since the investor is uncertain about the
probability measure she uses, she will seek for an optimal
investment proportion that is least affected by model uncertainty.
This leads to the so called \emph{robust forward performance
processes} as first introduced in \cite{Jan} and later extended to
the case with uncertain parameters in \cite{CL}.
{In \cite{Jan}
a class of convex penalty functions was introduced to represent the weighting/likehood of $\mathbb{P}^u$.
To generalize the idea of penalty functions, we parameterize robust
forward performance processes by the investor's prediction of the probability measure $\mathbb{P}^u$ (or the scenario parameter $u$), so penalty functions become a special class of parametrization. For this, we first give a precise meaning of parametrization which we will call \emph{a realization} of the model $\mathbb{P}^u$ (or the scenario parameter $u$) hereafter.}
{\begin{definition}\label{def081801}
For $0\leq t\leq s<\infty$, a mapping $\gamma:\Omega\times[t,s]\times
\mathcal{U}_{t,s}\rightarrow L^0(\mathcal{F}_s;\mathbb{R}^+)$ is a
realization process, if
for each $r$ with $t\leq r\leq s$ and $u\in\mathcal{U}_{t,s}$, it holds \begin{equation}\label{2021032001}
\gamma_{t,s}(u)=\gamma_{t,s}(u_1\oplus u_2)=\gamma_{t,r}(u_1)+\gamma_{r,s}(u_2),\  \text{a.s.,}
\end{equation}
where  $u_1$ and $u_2$ is the restriction of $u$ to trading interval $[t,r]$ and $[r,s]$, respectively, and we denote $u=u_1\oplus u_2$.
\end{definition}}
{We will use a realization process $\gamma_{t,s}(u)$ to parameterize the original utility on trading horizon $[t,s]$ because of the chosen model $\mathbb{P}^u$.
The condition \eqref{2021032001} is essentially a time-additivity property, which states that along the same model $\mathbb{P}^u$, the realization process on interval $[t,s]$  is accumulated by the realization processes estimated on $[t,r]$ and $[r,s]$.}

{The following generalises the definition of robust  forward
performance processes.}
\begin{definition}
\label{def}
{{A process $U\left( x,t\right) ,$ $\left( x,t\right) \in \mathbb{%
R_{+}\times }\left[ 0,\infty \right)$,
is a robust forward
performance process associated with  a realization process
$\gamma$ and a parameter $\tau\in\mathbb{R}$
if}}

{{{\rm i)}\ for each $x\in \mathbb{R}_{+},$ $U\left( x,t\right) $ is $\mathbb{F}$%
-progressively measurable;}}

{{{\rm ii)}\ for each $t\geq 0$, the
mapping $x\mapsto U(x,t)$ is strictly increasing and strictly
concave;}}

{{{\rm iii)}\ the process ${U}(x,t)$
satisfies the self-generating property (dynamic programming
principle), i.e., for all $s\geq t\geq 0$ and $\bar{u}\in\mathcal{U}_{0,t}$,}
\begin{equation}\label{maxmim-1}
{\esssup_{\pi\in\tilde{\Pi}}\essinf_{u\in\mathcal{U}_{t,s}}E_{\mathbb{P}^u}[\widetilde{U}(X_s^{\pi},s,\bar{u}\oplus u)|\mathcal{F}_t,X_t^\pi=x]=\widetilde{U}(x,t,\bar{u}),\
\ \ a.s.}
\end{equation}
where
\begin{equation}\label{2021032002}
\widetilde{U}\left( x,t,u\right)=U(x,t)+\tau\gamma_{0,t}(u).
\end{equation}}
\end{definition}
\begin{remark}\label{re-def1}
{
It is easy to check that the process $\widetilde{U}$ defined in \eqref{2021032002} also satisfies properties ${\rm i)}-{\rm ii)}$ in Definition \ref{def}
and thus $\widetilde{U}$ parameterized by the scenario parameter $u$ can also be regarded as a robust forward performance process, where $\gamma_{0,t}(u)$ represents a realization of the model $\mathbb{P}^u$ up to time $t$.
In this situation, the robust forward performance process $\widetilde{U}(x,t,u)$ at any given time $t$ constitutes of the original utility $U(x,t)$ and the realization process $\gamma_{0,t}(u)$ reflecting the historical cumulative impact of the model $\mathbb{P}^u$ from $0$ to $t$.
Considering the utility may be increasing or decreasing along with the model $\mathbb{P}^u$ which depends on the investor's attitude to this model,  we introduce a parameter $\tau$ in \eqref{2021032002} with its sign indicating the varying trend of the utility and its absolute value $|\tau|$ modeling the sensitively of the utility with respect to the model $\mathbb{P}^u$ (or the scenario $u$.)}
\end{remark}

\begin{remark}\label{re-def2}
{It turns out that our definition of robust forward performance
processes is a generalization of robust forward performance processes studied in
\cite{Jan}, although they are proposed based on different ideas. In fact,
it is easy to check from \eqref{2021032001} that property \eqref{maxmim-1} is equivalent to the following form
\begin{equation}\label{maxmim}
{\esssup_{\pi\in\tilde{\Pi}}\essinf_{u\in\mathcal{U}_{t,s}}E_{\mathbb{P}^u}[U(X_s^{\pi},s)
+\tau\gamma_{t,s}(u)|\mathcal{F}_t,X_t^\pi=x]=U(x,t),\
\ \ a.s.}
\end{equation}
If $\tau\geq 0$ and the realization process $\gamma_{t,s}(u)$ is convex in $u$, then $\tau\gamma_{t,s}(u)$ becomes a penalty function.
{
In this situation, we can verify from \eqref{maxmim} that our
Definition \ref{def} is consistent with \cite{Jan}.
Due to
the convexity of $\gamma_{t,s}(u)$ on $u$, a duality method is
developed in \cite{Jan} to construct $U(x,t)$ and its associated
optimal investment proportion $\pi^*$. Moreover, a saddle point method is employed to further find the
worst case scenario $u^*$ in \cite{CL} (with
$\tau= 0$).}}

{
}

{On the other hand, if $\tau<0$, the economic interpretation of $\tau\gamma_{0,t}(u)$ in \eqref{2021032002} is as follows.
From \eqref{2021032002} the investor's utility at time $t$ includes two terms $U(x,t)$ and $\tau\gamma_{0,t}(u)$
when using the model $\mathbb{P}^u$. Since $\tau$ is negative, the utility  $U(x,t)+\tau\gamma_{0,t}(u)$ is smaller than the original utility $U(x,t)$ and
the term $|\tau|\gamma_{0,t}(u)$ can be regarded as the utility's loss from time $0$ to $t$ under model $\mathbb{P}^u$.
 As a result, the  utility at time $t$ becomes $U(x,t)+ \tau\gamma_{0,t}(u)$. 
Following this viewpoint, the term $\tau\gamma_{0,t}(u)$ represents the degree of utility loss based on the investor's prediction over the probability measure $\mathbb{P}^u$ stemming from model uncertainty. }

{We remark that we do not assume
$\gamma_{t,s}(u)$ is convex with respect to $u$, which implies that the duality
method may not be suitable to our framework even for $\tau>0$. In addition, the
classical saddle point argument used in \cite{CL} is not valid either
because the saddle point in general does not exist (for example, when the term $\tau\gamma_{t,s}(u)$ is concave in $u$).}
\end{remark}




\subsection{{A stochastic differential game approach}}

\label{subsection2.1}

In contrast to \cite{Jan} and \cite{CL}, which do not cover the case
of the stochastic factor model, we aim to construct a class of
robust forward performance processes with explicit dependency on the
stochastic factor process $V$. Our approach is based on
\emph{stochastic differential games}, which is an alternative and
more direct approach for the case $\tau= 0$, and may also work for
the case $\tau\neq 0$ when the realization process $\gamma$ has some
specific forms (see Section 7 for the case $\tau<0$).

The basic idea of the stochastic differential game approach is as
follows. To robustify the optimal investment, the inner part of the
above optimization problem (\ref{maxmim}) is played by the market minimizes the expected forward
utility by choosing the worst-case scenario, whereas the investor
aims to select the best investment proportion that is least affected
by the market's choice. This leads to a stochastic
differential game between the investor and market.

In addition to the representation of the robust forward performance process, we also aim to provide both the optimal investment proportion for each scenario and the worst-case scenario for each investment proportion. The investment proportion (resp. worst-case scenario)
responding to each scenario (resp. investment proportion) can be exactly expressed as the  ``strategy to control"  in the setup of stochastic differential games (see \cite{BL,FS1989}). Thus, we next give the definitions of two admissible ``strategies" associated with their respective ``controls".
\begin{definition}
An admissible investment strategy  responding to each scenario parameter for an investor  is a
 mapping $\alpha:[0,\infty)\times\Omega\times\mathcal{U}\rightarrow\tilde{\Pi}$ satisfying the following two properties:\\
{\rm i)} For each $u\in\mathcal{U}$, $\alpha$ is $\mathbb{F}$-progressively measurable;\\
{\rm ii)}  Non-anticipative property, that is, for all $t>0$ and all
$u_1,u_2\in\mathcal{U}$, with $u_1=u_2,$ $dsd\mathbb{P}$-a.e.,  on
$[0,t]$, it holds that $\alpha(\cdot,u_1)=\alpha(\cdot,u_2)$,
$dsd\mathbb{P}$-a.e.,  on $[0,t]$.

An admissible scenario parameter strategy responding to each
investment proportion for the market,
$\beta:[0,\infty)\times\Omega\times\tilde{\Pi}\rightarrow\mathcal{U}$,
is defined similarly. The set of all admissible investment
strategies for the investor is denoted by $\mathcal{A}$, while the
set of all admissible scenario parameter strategies  is denoted by
$\mathcal{B}$.
\end{definition}

Herein, the concept \emph{non-anticipative property} is widely used in the definition of \emph{admissible strategies} in
differential games to characterize that each player's decision,
depending on the other's action, will not change if the other one
chooses the same control (see \cite{BL}). We introduce this concept
here to enforce that an investor will take the same investment
action if the scenario does not change.

We consider a zero-sum stochastic differential game, where the state
dynamic is given by the wealth equation \eqref{wealth-process}.
Furthermore, let $U(x,t)$ be a stochastic process satisfying i) and
ii) in Definition \ref{def} and $\gamma$ be a realization process
 with parameter $\tau\in\mathbb{R}$. For any $s\geq t$,
the objective functional is given by
$${J(x,t;s,\pi,u)=E_{\mathbb{P}^u}[U(X_s^\pi,s)+\tau\gamma_{t,s}(u)|\mathcal{F}_t, X_t^\pi=x].}$$
The lower and upper values of the game are then defined as
\begin{equation}\label{maxmin_1}
{\underline{U}(x,t;s)=\essinf_{\beta\in\mathcal{B}}\esssup_{\pi\in\tilde{\Pi}}J(x,t;s,\pi,\beta(\cdot,\pi)),
\ \ \ a.s.,}
\end{equation}
and
\begin{equation}\label{minmax}
{\overline{U}(x,t;s)=\esssup_{\alpha\in\mathcal{A}}\essinf_{u\in
\mathcal{U}} J(x,t;s,\alpha(\cdot,u),u), \ \ \ a.s.,}
\end{equation}
respectively.

Note that if $\underline{U}(x,t;s)=U(x,t)$, for all $s\geq t$, which
implies that the objective functional of the stochastic differential
game ``self generates" the lower value of the game, then it is clear
that $U(x,t)$ becomes a robust forward performance process
satisfying i)-iii) in Definition \ref{def}. Thus, we say the game is
\emph{self-generating}  if
\begin{equation}\label{2020083101} \underline{U}(x,t;s)=U(x,t),\
\text{for\ all}\ s\geq t,
 \end{equation} which will in turn provide a robust forward performance process.
To this end, we will construct a control $\pi^*\in\tilde{\Pi}$, a
strategy $\beta^*\in\mathcal{B}$, and a process $U(x,t)$ satisfying
the martingale properties: For any $\pi\in\tilde{\Pi}$, $s\geq t$,
a.s.,
\begin{equation}\label{beta_star}
\essinf_{\beta\in\mathcal{B}}J(x,t;s,\pi,\beta(\cdot,\pi))=J(x,t;s,\pi,\beta^*(\cdot,\pi))\leq
U\left(x,t\right);
\end{equation}
\begin{equation}\label{martingale_maxmin}
J(x,t;s,\pi^*,\beta^*(\cdot,\pi^*)) =U\left(x,t\right).
\end{equation}

\subsection{{Further discussion on the stochastic differential game approach when the game value exists}}

\label{subsection2.2}

In this subsection, according to the sign of the parameter $\tau$ in
\eqref{maxmim}, we further explain the application of our stochastic
differential game method in robust forward investment problems.

If $\tau\geq 0$ and $\gamma_{t,s}(u)$ is convex in $u$,
a saddle point in general exists at least for a special class of
penalty functions $\gamma_{t,s}(u)$ (as shown in \cite{Jan} and \cite{CL}). In contrast
to the saddle point method, the advantage of the stochastic
differential game approach is to provide, in explicit form, the
optimal investment choice for the investor not only under the
worst-case scenario but also for each scenario, as well as the worst
case scenario for each investment choice not only the optimal one.
Moreover, it is often relatively easy to compute the optimal
strategy pair $(\alpha^*,\beta^*)$, as they only involve
maximization/minimization problems rather than maxmin/minmax
problems.

\emph{From Sections 3 to 6, we consider a robust forward problem
without realization process, i.e., $\tau= 0$.} In this situation, the game
value exists, and we will construct the associated forward
performance process by the value of the stochastic differential
game. Recall that the value of the game exists if
\begin{equation}\label{game_vlaue}
\underline{U}(x,t;s)=\overline{U}(x,t;s),\ \text{for\ all}\ s\geq
t,
\end{equation}
 which further implies that both equal to $U(x,t)$ if the
self-generating condition \eqref{2020083101} also holds. In this
situation, with the help of the upper value function
$\overline{U}(x,t;s)$,
 we will construct a control pair $(\pi^*,u^*)\in\tilde{\Pi}\times\mathcal{U}$, a strategy pair $(\alpha^*,\beta^*)\in\mathcal{A}\times\mathcal{B}$, and a process $U(x,t)$ satisfying the martingale properties: \eqref{beta_star}, \eqref{martingale_maxmin}, and
for any $u\in\mathcal{U}$, $s\geq t$, a.s.,
\begin{equation}\label{alpha_star}
\esssup_{\alpha\in\mathcal{A}}
J(x,t;s,\alpha(\cdot,u),u)=J(x,t;s,\alpha^*(\cdot,u),u)\geq U(x,t);
\end{equation}
\begin{equation}
J(x,t;s,\alpha^*(\cdot,u^*),u^*)={U}\left(x,t\right).\label{martingale_minmax}
\end{equation}
Note that (\ref{beta_star}) and (\ref{martingale_maxmin}) are the martingale characterization of the lower value of the game in (\ref{maxmin_1}), whereas (\ref{alpha_star}) and (\ref{martingale_minmax}) characterize the upper value of the game in (\ref{minmax}).

Moreover, if it also holds that $\pi^*=\alpha^*(\cdot,u^*)$ and
$u^*=\beta^*(\cdot,\pi^*)$, then the martingale conditions
(\ref{beta_star})-\eqref{martingale_maxmin} and \eqref{alpha_star}-(\ref{martingale_minmax}) further imply that, for
all $s\geq t$,
\begin{align*}
J(x,t;s,\pi^*,u)&\geq J(x,t;s,\pi^*,\beta^{*}(\cdot,\pi^*))\\
&=J(x,t;s,\alpha^*(\cdot,u^*),u^*)\geq J(x,t;s,\pi,u^*),
\end{align*}
so the control pair $(\pi^*,u^*)$ is a saddle point for the
stochastic differential game with the value
$J(x,t;s,\pi^*,u^*)=U(x,t)$.

\begin{remark}
When $\tau\neq 0$, it is unclear how to construct a robust forward
performance process with a general realization process $\gamma$ (see,
for example, the discussion of the time consistency issue of penalty
functions in Section 4 of \cite{Jan}). Moreover, a saddle point may
even fail to exist if $\tau<0$. Nevertheless, we will show in
section 7 that our stochastic differential game approach may still
work for $\tau<0$, at least for a special class of {quadratic form realization processes}. A more general case for $\tau\neq 0$ is still left open.
\end{remark}
%
%

\section{Power robust  forward performance processes with zero realization processes}

\label{SectionPowerUtility}
In this section, we focus on a class of homothetic robust forward performance processes that are homogenous in the degree of $\delta\in(0,1)$, and has the factor-form%
\begin{equation}
U\left( x,t\right) =\frac{x^{\delta }}{\delta }e^{f(V_{t},t)},
\label{power-general}
\end{equation}%
where $f:\mathbb{R}^{d}\times \lbrack 0,\infty )\rightarrow \mathbb{R}$ is a
deterministic function to be specified, and the parameter $\tau$ of the realization process is equal to zero.
We call such a robust forward performance process a \emph{power robust forward performance process}.

\begin{proposition}\label{le-PDE}
 Assume that $f(v,t)$,
$\left( v,t\right) \in \mathbb{R}^{d}\times
\lbrack 0,\infty ),$  is a classical solution (with enough regularity) of the semilinear\textbf{\ }PDE
\begin{equation}
f_{t}+\frac{1}{2}Trace\left( \kappa \kappa ^{T}\nabla ^{2}f\right) +\eta
(v)^{T}\nabla f+G(v,\kappa ^{T}\nabla f)=0,  \label{f-eqn}
\end{equation}%
where
\begin{equation}
G(v,z)=\inf_{u\in U}\sup_{\pi\in\Pi}F(v,z,\pi,u),
\label{driver-formal}
\end{equation}%
with
\begin{equation}\label{112001}
F(v,z,\pi,u)=-\frac{1}{2}\delta(1-\delta)|\pi|^2
+\delta\pi^T(\theta(v)+z+u)+z^Tu+\frac{1}{2}|z|^2.
\end{equation}
Then, $U(x,t)=\frac{x^{\delta }}{\delta }e^{f(V_{t},t)}$ is a power robust forward performance process.
\end{proposition}

\begin{proof}
Since $U(x,t)$ obviously satisfies i) and ii) in Definition  \ref{def}, it is sufficient to examine iii) in  Definition \ref{def}.

\emph{Step 1.} From \eqref{driver-formal} and \eqref{112001}, we have
\begin{equation}\label{111902}
G(v,z)=\inf_{u\in U}\sup_{\pi\in\Pi}F(v,z,\pi,u)=\inf_{u\in U}F(v,z,\alpha^*(v,z,u),u),
\end{equation}
with
\begin{equation}\label{optimal-1}
\alpha^*(v,z,u)=argmax_{\pi\in\Pi}F(v,z,\pi,u)=Proj_{\Pi}(\frac{\theta(v)+z+u}{1-\delta}).
\end{equation} Using the Lipschitz
continuity of the projection operator on the convex set $\Pi$, there exists a Borel measurable mapping
$u^*:\mathbb{R}^d\times\mathbb{R}^d\rightarrow U$ such that
\begin{equation}\label{optimal-2}
u^*(v,z)=argmin_{u\in U}F(v,z,\alpha^*(v,z,u),u).
\end{equation}
Then, from \eqref{111902} and \eqref{optimal-2}, we have
\begin{equation}
G(v,z)=F(v,z,\pi^*(v,z),u^*(v,z)),
\label{111801}
\end{equation}
 with
\begin{equation}\label{optimal-3}
 \pi^*(v,z):=\alpha^*(v,z,u^*(v,z)).
 \end{equation}

We claim that, for any $u\in U$,
\begin{equation}\label{key}
F(v,z,\pi^*(v,z),u)\geq F(v,z,\pi^*(v,z),u^*(v,z)).
\end{equation}
If (\ref{key}) holds, then
\begin{align}
\sup_{\pi\in\Pi}\inf_{u\in U}F(v,z,\pi,u)&\geq\inf_{u\in U}F(v,z,\pi^*(v,z),u)\label{first}\\
&\geq F(v,z,\pi^*(v,z),u^*(v,z))\label{second}\\
&= G(v,z)=\inf_{u\in U}\sup_{\pi\in\Pi}F(v,z,\pi,u),\nonumber
\end{align}
so both (\ref{first}) and (\ref{second}) become equalities. In turn, $\pi^*(v,z)$ in (\ref{optimal-3}) and $u^*(v,z)$ in (\ref{optimal-2}) satisfy, respectively,
\begin{equation*}
\pi^*(v,z)=argmax_{\pi\in\Pi}\inf_{u\in U}F(v,z,\pi,u),
\end{equation*}
and
\begin{equation}\label{third}
u^*(v,z)=argmin_{u\in U}F(v,z,\pi^*(v,z),u).
\end{equation}

On the other hand, there exists a $U$-valued Borel measurable mapping $\bar{\beta}^*(v,z,\pi)$ such that $F(v,z,\pi,u)$ attains the minimum, i.e.
$$\inf_{u\in U}F(v,z,\pi,u)=F(v,z,\pi,\bar{\beta}^*(v,z,\pi)).$$
Then, from (\ref{third}), the mapping $\beta^*(v,z,\pi)$ defines as
\begin{equation}\label{optimal-4}
\beta^*(v,z,\pi)=\left\{
\begin{array}{ll}
u^*(v,z),\ &\text{if}\ \pi=\pi^*(v,z);\\
\bar{\beta}^*(v,z,\pi),\ &\text{otherwise,}
\end{array}
\right.
\end{equation}
also minimizes $F(v,z,\pi,u)$ over $u\in U$, and moreover,
\begin{equation}
\pi^*(v,z)=argmax_{\pi\in\Pi}F(v,z,\pi,\beta^*(v,z,\pi)).
\end{equation}

\emph{Step 2.} We are left to prove the inequality (\ref{key}). We omit the variables $(v,z)$ in $\pi^*(v,z)$ and $u^*(v,z)$, and write them as $\pi^*$ and $u^*$ in this step. For any $u\in U$ and $\lambda\in(0,1)$ let
$$u_1:=\lambda u+(1-\lambda) u^*.$$
Set $\pi_1:=\alpha^*(v,z,u_1)$ and recall from (\ref{optimal-3}) that $\pi^*=\alpha^*(v,z,u^*)$. Then, it follows from (\ref{optimal-2}) that
\begin{align*}
F(v,z,\pi^*,u^*)&\leq F(v,z,\pi_1,u_1)\\
&=\lambda F(v,z,\pi_1,u)+(1-\lambda)F(v,z,\pi_1,u^*)\\
&\leq \lambda F(v,z,\pi_1,u)+(1-\lambda) F(v,z,\pi^*,u^*).
\end{align*}
where we used
$F(v,z,\pi,u)\leq F(v,z,\alpha^*(v,z,u),u)$ in the last inequality. Thus, $$F(v,z,\pi^*,u^*)\leq F(v,z,\pi_1,u)=F(v,z,\alpha^*(v,z,u_1),u)$$ for any $u\in U$. Sending $\lambda\rightarrow 0$ and using the continuity of $\alpha^*(v,z,u)$ in $u$, we have $\alpha^*(v,z,u_1)\rightarrow \alpha^*(v,z,u^*)=\pi^*$. Then, the inequality (\ref{key}) follows by the continuity of $F(v,z,\pi,u)$ in $\pi$.

\emph{Step 3.}
Using the homothetic form \eqref{power-general} and applying It\^o's formula to $U(X_s^\pi,s)$,  we get
\begin{equation*}
\begin{aligned}
&dU(X_s^\pi,s)\\
=&\ U(X_s^\pi,s)\big[f_{s}+\frac{1}{2}Trace\left( \kappa \kappa ^{T}\nabla ^{2}f\right) +\eta
(V_s)^{T}\nabla f+F(V_s,\kappa ^{T}\nabla f,\pi_s,u_s)\big]ds\\
&+U(X_s^\pi,s)(\delta\pi_s^T+\nabla f^T\kappa)dW_s^u.
\end{aligned}
\end{equation*}
For any $s\geq t\geq 0$, from \eqref{f-eqn}, we further get
\begin{equation}\label{111901}
\begin{aligned}
&E_{\mathbb{P}^u}[U(X_s^\pi,s)|\mathcal{F}_t,X_t^\pi=x]-U(x,t)\\
=&\ J(x,t;s,\pi,u)-U(x,t)\\
=&\ E_{\mathbb{P}^u}\big[\int_t^sU(X_r^\pi,r)\big(F(V_r,\kappa ^{T}\nabla f,\pi_r,u_r)-G(V_r,\kappa ^{T}\nabla f)\big)dr|\mathcal{F}_t,X_t^\pi=x\big].
\end{aligned}
\end{equation}
We set
\begin{equation}\label{112301000}
\begin{aligned}
&\pi^*_t=\pi^*(V_t,\kappa\nabla f(V_t,t)),\ u^*_t=u^*(V_t,\kappa\nabla f(V_t,t)),\\
&\alpha^*(t,u_t)=\alpha^*(V_t,\kappa\nabla f(V_t,t),u_t),\ \beta^*(t,\pi_t)=\beta^*(V_t,\kappa\nabla f(V_t,t),\pi_t),
\end{aligned}\end{equation}
with the mappings $(\pi^*,u^*,\alpha^*,\beta^{*})$  given in (\ref{optimal-3}), \eqref{optimal-2}, \eqref{optimal-1} and \eqref{optimal-4}, respectively.
Then, it is easy to check that $U(x,t)$  satisfies the martingale conditions (\ref{beta_star})-\eqref{martingale_maxmin} and \eqref{alpha_star}-(\ref{martingale_minmax}), which implies that $U(x,t)=\frac{x^{\delta }}{\delta }e^{f(V_{t},t)}$ is a power robust forward performance process, with the optimal control pair $(\pi^*,u^*)$ and the optimal strategy pair $(\alpha^*,\beta^*)$.
\end{proof}

\begin{remark}
It is worth to point out that the strategies $\alpha^*$ and $\beta^*$ we constructed in the above proof
are also called ``counterstrategies"; the reader can refer to Chapter 10, Section 1 in \cite{KS} for more details.

Since, by our construction, $\pi_t^*=\alpha^*(t,u_t^*)$ and
$u_t^*=\beta^*(t,\pi_t^*)$, it follows that $(\pi^*,u^*)$ is
actually a saddle point for the associated game. However, compared
to the classical saddle point argument such as Sion's Minimax
Theorem (see, for example, \cite{CL, YLZ}), our formulae are more
explicit and is constructed via their corresponding
counterstrategies.
\end{remark}

{Note that the semi-linear PDE \eqref{f-eqn} is a new class of
Hamilton-Jacobi-Bellman-Isaacs equations, which is ill-posed for the
equation is posed forward in time. Due to this ``wrong" time
direction, one does not expect solutions to exist for all initial
conditions or to depend continuously on them, making the problem
ill-posed.} A similar difficulty also appears in \cite{ASS},
\cite{NT2017}, \cite{NZ} and \cite{SSZ} for the construction of
forward processes without model ambiguity, where the Widder's
theorem is employed. Nevertheless, the form of PDE \eqref{f-eqn}
motivates us how to construct the optimal investment proportion,
worst-case scenario parameter and the related optimal strategies for
different situations, which will be used in the following Theorem
\ref{Theorem2_ForwardUtility}. In order to  give the specific form
of the process $f\left(V_{t},t\right) $, we bypass PDE \eqref{f-eqn}
by directly using  the Markovian solution of an ergodic BSDE whose
driver has the form \eqref{driver-formal}. This approach was first
introduced in \cite{LZ} to study the forward performance process in
the absence of model uncertainty. We first give the existence and
uniqueness of the Markovian solution of the associated ergodic BSDE.
{For this, we further strengthen Assumption 1 by requiring
the constant $C_\eta$ given in
\eqref{dissipative} satisfies
\begin{equation}\label{C_eta}
C_{\eta}\geq \frac{3\delta C_{\theta}}{1-\delta}\left[(K_{\theta}+K_u)\vee 1\right],
\end{equation}
where $\delta\in(0,1)$ is the risk aversion degree, $C_{\theta}$ and $K_{\theta}$ are the Lipschitz constant and bound of the market price of risk vector $\theta(v)$ respectively, and $K_u$ is the bound of the scenario parameter $u$ with $K_u=\max_{u\in U}|u|$.}

\begin{lemma}\label{le112301}
Assume the function $G$ has the form \eqref{driver-formal}. Then, the ergodic BSDE
\begin{equation}
dY_{t}=(-G(V_{t},Z_{t})+\lambda )dt+Z_{t}^{T}dW_{t},  \label{EQBSDE1}
\end{equation}%
admits a unique Markovian solution $(Y_{t},Z_{t},\lambda ),$ $t\geq 0,$ i.e.,
there exist a unique constant $\lambda$ and functions $y:%
\mathbb{R}^{d}\rightarrow \mathbb{R}$, $z:\mathbb{R}^{d}\rightarrow
\mathbb{R}^{d}$ such that $ Y_{t}=y\left(
V_{t}\right) ,\ Z_{t} =z\left( V_{t}\right)$. Here, we say a Markovian solution is unique in the following sense: the function $y(\cdot )$ is
unique up to a constant and has at most linear growth, and $z(\cdot )$ is
bounded.
\end{lemma}
\begin{proof}
Using the Lipschitz continuity of the projection operator, it follows from \eqref{112001} and \eqref{optimal-1}
that
\begin{equation}\label{112302}
\begin{aligned}
&|F(v,z,\alpha^*(v,z,u),u)-F(\bar{v},z,\alpha^*(\bar{v},z,u),u)|\leq C(1+|z|)\cdot|v-\bar{v}|,\\
&|F(v,z,\alpha^*(v,z,u),u)-F(v,\bar{z},\alpha^*(v,\bar{z},u),u)|\leq C(1+|z|+|\bar{z}|)\cdot|z-\bar{z}|,\\
&|F(v,0,\alpha^*(v,0,u),u)|\leq C.\\
\end{aligned}
\end{equation}
{Indeed, to show the first inequality, we note from \eqref{112001} that
\begin{align*}
     &|F(v,z,\alpha^*(v,z,u),u)-F(\bar{v},z,\alpha^*(\bar{v},z,u),u)|\\
\leq &\ \frac{\delta(1-\delta)}{2}\left|\alpha^*(v,z,u)+\alpha^*(\bar{v},z,u)\right|\times\left|\alpha^*(v,z,u)-\alpha^*(\bar{v},z,u)\right| \\
&+\delta|\theta(v)+z+u|\times\left|\alpha^*(v,z,u)-\alpha^*(\bar{v},z,u)\right|+\delta|\alpha^*(\bar{v},z,u)|\times|\theta(v)-\theta(\bar{v})|.
\end{align*}
Since the projection operator $Proj_{\Pi}(\cdot)$ is Lipschitz continuous with its Lipschitz constant $1$ and $0\in {\Pi}$, from \eqref{optimal-1} we have
$$\left|\alpha^*(v,z,u)-\alpha^*(\bar{v},z,u)\right|\leq \frac{1}{1-\delta}|\theta(v)-\theta(\bar{v})|\leq \frac{C_{\theta}}{1-\delta}|v-\bar{v}|,$$
and
$$|\alpha^*(\bar{v},z,u)|\leq \frac{1}{1-\delta}|\theta(\bar{v})+z+u|\leq \frac{K_{\theta}+|z|+K_u}{1-\delta}.$$
In turn,
\begin{align*}
     &|F(v,z,\alpha^*(v,z,u),u)-F(\bar{v},z,\alpha^*(\bar{v},z,u),u)|\\
\leq &\ 3\frac{\delta}{1-\delta}(K_{\theta}+|z|+K_u)C_{\theta}|v-\bar{v}|\\
\leq&\  \frac{3\delta C_{\theta}}{1-\delta}\left[(K_{\theta}+K_u)\vee 1\right](1+|z|)\cdot|v-\bar{v}|\leq C_{\eta}(1+|z|)\cdot|v-\bar{v}|,
\end{align*}
with $C_{\eta}$ given in (\ref{C_eta}). The other two inequalities in (\ref{112302}) can be proved in a similar way. Furthermore, we note that the constant $C$ in (\ref{112302}) is independent of $u\in U$. Hence, from \eqref{111902}, we further obtain
\begin{equation}\label{120401}
\begin{aligned}
&|G(v,z)-G(\bar{v},z)|\leq C(1+|z|)\cdot|v-\bar{v}|,\\
&|G(v,z)-G(v,\bar{z})|\leq C(1+|z|+|\bar{z}|)\cdot|z-\bar{z}|,\
|G(v,0)|\leq C.
\end{aligned}
\end{equation}
Therefore, from Proposition 3.1 and Appendix A in \cite{LZ} we obtain the desired result.}
\end{proof}

{Roughly speaking, the additional ``large enough"\ requirement of the constant $C_{\eta }$ in (\ref{C_eta}) is to guarantee the forward stochastic factor process $V$ converges fast enough to dominate the dissipative nature of the backward equation for $Y$ reflected by the Lipschitz constant $C$ in the first inequality of (\ref{120401}).  This additional requirement plays an important role in the study of ergodic BSDE \eqref{EQBSDE1}. We refer the reader to Appendix A in \cite{LZ}.}





We next present the specific form of the process $f(V_t,t)$ by using the solution of
the ergodic BSDE \eqref{EQBSDE1}.

\begin{theorem}
\label{Theorem2_ForwardUtility}

Let $(Y_{t},Z_{t},\lambda )=(y(V_{t}),z(V_{t}),\lambda ),t\geq 0,$ be the
unique Markovian solution of (\ref{EQBSDE1}). Then,
 the process $U(x,t),$ $\left( x,t\right) \in \mathbb{R}_{+}\times \left[
0,\infty \right) ,$ given by
\begin{equation}
U(x,t)=\frac{x^{\delta }}{\delta }e^{y(V_{t})-\lambda t}\text{ ,}
\label{PowerForwardUtility}
\end{equation}%
is a power robust forward performance process.
Moreover, the optimal portfolio weight $\pi^*$, the worst-case scenario parameter $u^*$ and
the optimal  strategies $\alpha^*,\beta^*$ responding to each scenario parameter $u$ and portfolio weight $\pi$ are given as follows
\begin{equation}\label{112301}
\begin{aligned}
&\pi^*_t=\pi^*(V_t,z(V_t)),\ u^*_t=u^*(V_t,z(V_t)),\\
&\alpha^*(t,u_t)=\alpha^*(V_t,z(V_t),u_t),\ \beta^*(t,\pi_t)=\beta^*(V_t,z(V_t),\pi_t),
\end{aligned}\end{equation}
where the mappings $(\pi^*,u^*,\alpha^*,\beta^{*})$ are given in \eqref{optimal-3}, \eqref{optimal-2}, \eqref{optimal-1} and \eqref{optimal-4}, respectively.

In addition, the associated wealth process $X^*$ under the worst-case scenario is given by
$$X^*_t=X_0\mathcal{E}\Big(\int_0^t(\pi^*_s)^T\cdot[(\theta(V_s)+u_s^*)ds
+dW_s^{u^*}]\Big).$$

\end{theorem}
{It is worth to point out that the constant $\lambda$, as a part of the solution of ergodic BSDE \eqref{EQBSDE1}, can be regarded as
the optimal
long-term growth rate of the corresponding expected utility of wealth with model
uncertainty (see Remark \ref{20200921}).}

{We now give the proof of Theorem \ref{Theorem2_ForwardUtility}.}

\begin{proof}
It is easy to check that the process given by \eqref{PowerForwardUtility} is $\mathbb{F}$-progressively measurable, strictly increasing and strictly concave in $x$. We only need to show that the martingale conditions (\ref{beta_star})-\eqref{martingale_maxmin} and \eqref{alpha_star}-(\ref{martingale_minmax})
hold. For this, from \eqref{wealth-process}, \eqref{new-B.M.} and \eqref{EQBSDE1} we get,
for all $s\geq t\geq 0$, $(\pi,u)\in\tilde{\Pi}\times\mathcal{U}$,
\begin{equation*}
\begin{aligned}
X_s^\pi&=X_t^\pi\cdot\exp\Big\{\int_t^s\pi_r^T(\theta(V_r)+u_r)-\frac{1}{2}|\pi_r|^2dr
+\int_t^s\pi_r^TdW_r^u\Big\},\\
(Y_s-\lambda s)&=(Y_t-\lambda t)-\int_t^sG(V_r,Z_r)-Z_r^Tu_rdr+\int_t^sZ_r^TdW_r^u.
\end{aligned}
\end{equation*}
Thus, we have
\begin{equation*}
\begin{aligned}
&U(X_s^\pi,s)=\frac{(X_s^\pi)^\delta}{\delta}e^{Y_s-\lambda s}\\
=&\ U(X_t^\pi,t)\cdot\mathcal{E}\Big(\int_t^s(\delta\pi^T_r+Z_r^T)dW_r^u\Big)
\cdot\exp\Big\{\int_t^sF(V_r,Z_r,\pi_r,u_r)-G(V_r,Z_r)dr\Big\}.
\end{aligned}
\end{equation*}
Therefore,%
\begin{equation*}
\begin{aligned}
&E_{\mathbb{P}^u}\left[ U(X_s^\pi,s)|\mathcal{F}%
_{t},X_t=x\right]-U(x,t)\\
=&\ J(x,t;s,\pi,u)-U(x,t)\\
=&\ U(x,t)\cdot E_{\mathbb{P}^u}\left(\frac{M_s}{M_t}\cdot
\exp\Big\{\int_t^s\Big(F(V_r,Z_r,\pi_r,u_r)-G(V_r,Z_r)\Big)dr\Big\} \Big\vert \mathcal{F}_{t}\right)-U(x,t),
\end{aligned}
\end{equation*}%
where, for $ t\in [0,s]$, $
M_t:=\mathcal{E}\Big(\int_0^t(\delta\pi^T_r+Z_r^T)dW_r^u\Big),
$
is a uniformly integrable exponential martingale (since $\pi$ satisfies the BMO-condition and $z(\cdot)$ is bounded). Similar to the argument in the proof of Lemma \ref{le-PDE},
we get the the martingale conditions (\ref{beta_star})-\eqref{martingale_maxmin} and \eqref{alpha_star}-(\ref{martingale_minmax}) from the above equality.
\end{proof}
\begin{remark}
The probability measure $\mathbb{P}^{u^*}$ associated with $u^*$ given in Theorem \ref{Theorem2_ForwardUtility} has the following form
\begin{equation}\nonumber
\frac{d\mathbb{P}^{u*}}{d\mathbb{P}}\Big|_{\mathcal{F}_t}=\mathcal{E}\left(\int_0^t(u_s^*)^TdW_s\right).
\end{equation}
Thus, as a byproduct, we obtain a specific formula for the least favorable martingale measure as considered in \cite{FG}.
\end{remark}

\begin{remark}
Similar to  Proposition 3.4 in \cite{LZ}, it is easy to check that
$$f(v,t)=y(v)-\lambda t$$
is a classical solution of the semilinear PDE \eqref{f-eqn} with the initial condition $f(v,0)=y(v)$, where $(y(V_t),z(V_t),\lambda)$ is the solution of ergodic BSDE \eqref{EQBSDE1}.
\end{remark}

Next, we build a connection between power robust forward performance
processes and the solutions of a family of infinite horizon BSDE.
For $\rho
>0$, we consider the following infinite horizon BSDE
\begin{equation}
dY_{t}^{\rho }=\left( -G(V_{t},Z_{t}^{\rho })+\rho Y_{t}^{\rho }\right)
dt+\left( Z_{t}^{\rho }\right) ^{T}dW_{t},  \label{IQBSDE1}
\end{equation}%
where the driver $G(\cdot ,\cdot )$ is given in (\ref{EQBSDE1}). Then, this BSDE
admits a unique Markovian solution  $\left( Y_{t}^{\rho },Z_{t}^{\rho }\right)=(y^{\rho}(V_t),z^\rho(V_t)) $. Moreover, there exists a subsequence,  denoted by $\rho_n$, such that
$$y(v)=\lim_{\rho_n\downarrow 0}y^{\rho_n}(v),\ z(v)=\lim_{\rho_n\downarrow 0}z^{\rho_n}(v),\ \lambda=\lim_{\rho_n\downarrow 0}{\rho_n} y^{\rho_n}(v_0),$$
where $(y(V_t),z(V_t),\lambda)$ is the solution of ergodic BSDE \eqref{EQBSDE1} and $v_0\in\mathbb{R}^d$ is an arbitrary given reference point.
These results were first obtained in \cite{HU2} with Lipschitz driver and then extended to the quadratic driver in \cite{LZ}.

Similar to the proof of Theorem \ref{Theorem2_ForwardUtility}, we can examine that the process $U^{\rho}(x,t)$ given by \eqref{PowerForwardUtility2} is still a power robust forward performance process and it converges in an appropriate discounted manner to the process $U(x,t)$ as $\rho$ tends to 0.
\begin{corollary}
 The process $U^{\rho }\left( x,t\right) ,$ $\left( x,t\right) \in \mathbb{%
R}_{+}\times \left[ 0,\infty \right) ,$ given by%
\begin{equation}
U^{\rho }(x,t)=\frac{x^{\delta }}{\delta }e^{y^{\rho }\left( V_{t}\right)
-\int_{0}^{t}\rho y^{\rho }\left( V_{s}\right) ds}
\label{PowerForwardUtility2}
\end{equation}%
is a power robust forward performance process and the optimal portfolio strategy $\alpha _{t}^{\ast ,\rho }$  for each scenario parameter $u$ is given by
\begin{equation*}
\alpha _{t}^{\ast ,\rho }(u)=Proj_{\Pi }\left( \frac{\theta
(V_{t})+z^{\rho }(V_{t})+u_t}{1-\delta }\right) .
\end{equation*}
Furthermore, there exists a
subsequence $\rho _{n}\downarrow 0$  such that, for $%
\left( x,t\right) \in \mathbb{R}_{+}\times \left[ 0,\infty \right) ,$
\begin{equation}
\lim_{\rho _{n}\downarrow 0}\frac{U^{\rho _{n}}(x,t)e^{-y^{\rho _{n}}\left(
v_{0}\right) }}{U(x,t)}=1.  \label{relationU}
\end{equation}%
and the associated optimal portfolio strategies $\alpha
^{\ast ,\rho _{n}}$ and $\alpha ^{\ast }$ satisfy
\begin{equation}
\lim_{\rho _{n}\downarrow 0}E_{\mathbb{P}}\int_{0}^{t}\left \vert \alpha
^{\ast ,\rho _{n}}(s,u_s)-\alpha ^{\ast }(s,u_s)\right \vert ^{2}ds=0,\ \text{for}\ t\geq 0,\ u\in\mathcal{U}.
\label{portfolios}
\end{equation}
\end{corollary}

{\begin{remark}\label{re-other
utilities} We have obtained the representation of the power robust
forward performance process in factor-form by combining the zero-sum
stochastic differential game and ergodic BSDE approach. In fact,
this approach can be applied to study other type of homothetic
robust forward performance processes, such as logarithmic and
exponential cases. More specifically, the processes $U^1(x,t)$ and
$U^2(x,t)$ given by
$$U^1(x,t)=\ln x+y^1(V_t)-\lambda^1 t,\ (x,t)\in\mathbb{R}_+\times[0,\infty),$$
$$U^2(x,t)=-e^{-\gamma x+y^2(V_t)-\lambda^2 t},\ (x,t)\in\mathbb{R}\times[0,\infty)$$
are  logarithmic and exponential (with risk aversion parameter $\gamma>0$) robust forward performance
processes, respectively, where
$(Y^i_t,Z^i_t,\lambda^i)=(y^i(V_t),z^i(V_t),\lambda^i)$, $i=1,2,\
t\geq 0$, are the unique Markovian solution of the ergodic BSDE
\eqref{EQBSDE1} with the generator $G=G_i,\ i=1,2$, respectively,
with
$$G_1(v,z)=\inf_{u\in U}\sup_{\pi\in\Pi}\{-\frac 1 2 |\pi|^2+\pi^T\theta(v)+(\pi^T+z^T)u\},$$
$$G_2(v,z)=\sup_{u\in U}\inf_{\pi\in\Pi}\{\frac 1 2 |\gamma \pi-z|^2-r\pi^T(\theta(v)+u)+z^Tu\}.$$
\end{remark}}

\section{{Examples}}\label{examples}
We apply Theorem  \ref{Theorem2_ForwardUtility} to analyze two
specific examples. The first example is driven by the Brownian noise
which can be fully hedged. The second example is a single stock
model correlated with a single stochastic factor where only partial
hedging is possible. In both examples, no constraints on
portfolios are imposed and optimal robust investment policies for
the power robust forward performance processes are given in the
feedback form of stochastic factors.

\subsection{Market model I}
We consider the case that the set $\Pi$ is large enough in the sense
that  the mappings $\alpha^*$ in \eqref{optimal-1} has the following
form
\begin{equation*}
\begin{aligned}
\alpha^*(v,z,u)&=Proj_{\Pi}(\frac{\theta(v)+z+u}{1-\delta})=\frac{\theta(v)+z+u}{1-\delta}.
\end{aligned}
\end{equation*}
Then, the mappings $u^*$ in \eqref{optimal-2} and $\pi^*$ in \eqref{optimal-3} as well as $\beta^*$ in \eqref{optimal-4} take the form
\begin{equation*}
\begin{aligned}
&u^*(v,z)=argmin_{u\in U}F(v,z,\alpha^*(v,z,u),u)=Proj_{U}\big(-\theta(v)-\frac{1}{\delta}z\big),\\
& \pi^*(v,z)=\alpha^*(v,z,u^*(v,z))=\frac{\theta(v)+z+Proj_{U}\big(-\theta(v)-\frac{1}{\delta}z\big)}{1-\delta},\\
&
\beta^*(v,z,\pi)=\left\{
\begin{array}{ll}
Proj_{U}\big(-\theta(v)-\frac{1}{\delta}z\big),\ &\text{if}\ \pi=\pi^*(v,z);\\
argmin_{u\in U}(\delta\pi+z)^Tu,\ &\text{otherwise.}
\end{array}
\right.
\end{aligned}
\end{equation*}
In this case, the ergodic BSDE \eqref{EQBSDE1} becomes
\begin{equation}\label{112101}
dY_{t}=\Big(-\frac{1}{2}\frac{\delta}{1-\delta}dist^2\big(U,-\theta(V_t)-\frac{1}{\delta}Z_t\big)+\frac{1}{2\delta}|Z_t|^2+Z_t^T\theta(V_t)+\lambda \Big)dt+Z_{t}^{T}dW_{t}.
\end{equation}%
In turn, from Theorem \ref{Theorem2_ForwardUtility}, we obtain the following result.
\begin{proposition}\label{prop-example1}
Denote by $(y(V_t),z(V_t),\lambda)$  the Markovian solution of \eqref{112101}. Then, the process  $U(x,t)$ given by
\begin{equation*}
U(x,t)=\frac{x^{\delta }}{\delta }e^{y(V_{t})-\lambda t}\text{ ,}
\end{equation*}%
is a power robust forward performance process.  Moreover, the optimal control pair  $(\pi^*,u^*)\in\tilde{\Pi}\times\mathcal{U}$ and optimal strategy pair $(\alpha^*,\beta^*)\in\mathcal{A}\times\mathcal{B}$, have the following feedback form
\begin{equation*}
\begin{aligned}
&\pi^*_t=\frac{\theta(V_t)+z(V_t)+Proj_{U}\big(-\theta(V_t)-\frac{1}{\delta}z(V_t)\big)}{1-\delta},\\ &u^*_t=Proj_{U}\big(-{\theta(V_t)-\frac{1}{\delta}z(V_t)}\big),\\
&\alpha^*(t,u_t)=\frac{\theta(V_t)+z(V_t)+u_t}{1-\delta},\\
&\beta^*(t,\pi_t)=\left\{
\begin{array}{ll}
Proj_{U}\big(-\theta(V_t)-\frac{1}{\delta}z(V_t)\big),\ &\text{if}\ \pi=\pi^*(v,z);\\
argmin_{u_t\in U}(\delta\pi_t+z(V_t))^Tu_t,\ &\text{otherwise.}
\end{array}
\right.
\end{aligned}
\end{equation*}
\end{proposition}
\begin{remark}\label{080701}
It is worth to point out that the presence of the uncertainty in our
forward setting may lead to extreme prediction and conservative
policy implications for an ambiguity-averse investor. In fact, if we
consider the situation that the set $U$ is large enough such that
$$u^*(v,z)=Proj_{U}\big(-\theta(v)-\frac{1}{\delta}z\big)=-\theta(v)-\frac{1}{\delta}z.$$
Then, ergodic BSDE \eqref{112101} has the form
\begin{equation}\label{072801}
dY_{t}=\Big(\frac{1}{2\delta}|Z_t|^2+Z_t^T\theta(V_t)+\lambda \Big)dt+Z_{t}^{T}dW_{t},
\end{equation}
and the robust optimal portfolio weight
$\pi^*_t=-\frac{1}{\delta}z(V_t)$. Note that $(0,0,0)$ is the unique
Markovian solution of ergodic BSDE \eqref{072801}, from Proposition
\ref{prop-example1} we get
\begin{equation*}
\begin{aligned}
&\pi^*_t=-\frac{1}{\delta}z(V_t)=0,\ u^*_t=-{\theta(V_t)-\frac{1}{\delta}z(V_t)}=-\theta(V_t),\\
&\alpha^*(t,u_t)=\frac{\theta(V_t)+u_t}{1-\delta},\\
&\beta^*(t,\pi_t)=\left\{
\begin{array}{ll}
-\theta(V_t),\ &\text{if}\ \pi=\pi^*(v,z);\\
argmin_{u_t\in U}(\delta\pi_t)^Tu_t,\ &\text{otherwise.}
\end{array}
\right.
\end{aligned}
\end{equation*}
This implies that the robust investment policy for an investor is no
actions to be taken in the market if the degree of the uncertainty
is too large for her. At the same time,  the worst-case scenario has
a simple form and  depends only on the market price of risk
$\theta(v)$ and the stochastic factor $V_t$.
{In addition, even if an investor is forced to
invest (or pursuit high profits) in some situations such as  the
investor has a wrong judgment on the uncertainty of the market, or
is influenced by other extreme events, the
optimal investment strategy $\alpha^*$ still gives the corresponding
action policy for different scenarios.}

\end{remark}

\subsection{Market model II}

We consider a single stock and single stochastic factor model.
In this situation, we suppose $n=1$
and $d=2$ in the state equations (\ref{stock-SDE}) and (\ref{factor-SDE}), i.e.,
\begin{equation}\label{120903}
\begin{aligned}
dS_{t}&=b(V_{t})S_{t}dt+\sigma (V_{t})S_{t}dW_{t}^{1}\text{,}\\
dV_{t}^{1}&=\eta (V_{t})dt+\rho dW_{t}^{1}+\sqrt{1-\rho^2}dW_{t}^{2}\text{
\  \ and \ }dV_{t}^{2}=0,
\end{aligned}
\end{equation}%
with constant $\rho\in(0,1)$ and $\sigma \left( \cdot \right) $ bounded by a positive
constant. Note that the stochastic factor cannot be traded directly so that the market model is typically incomplete. 

Here, we consider an optimal portfolio problem with no constraints,
i.e.,  $\Pi =\mathbb{R\times }\left \{ 0\right \} $ (which means
$\pi _{t}^{2}\equiv 0$). Let $U=\{(u_1,u_2): -R\leq u_1\leq u_2\leq
R\}$ (a triangle domain in $\mathbb{R}^2$) with some given constant
$R>0$.
Then, the wealth equation (\ref{wealth-process}) reduces to $%
dX_{t}^{\pi }=X_{t}^{\pi }\pi _{t}^{1}\left( \theta
(V_{t})dt+dW_{t}^{1}\right) $ with $\theta (V_{t})=b(V_{t})/\sigma (V_{t}),$
and the driver of (\ref{EQBSDE1}) takes the form
\begin{equation}\label{120901}
\begin{aligned}
G(v,z_{1},z_{2})
=& \frac{\delta}{2(1-\delta)}dist^2\Big([-R,R],-\theta(v)
-\frac{1}{\delta}z_1-\frac{1-\delta}{\delta}z_2I_{\{z_2\geq 0\}}\Big)
\\&-\frac{1}{2\delta}|z_{1}|^{2}-\theta(v)z_{1}
+\Big(\frac{2\delta-1}{2\delta}z_2-\frac{1}{\delta}z_1-\theta(v)\Big)z_2I_{\{z_2\geq 0\}}\\
&+
(\frac12 z_2+R)z_2I_{\{z_2<0\}}.
\end{aligned}
\end{equation}%
Then, from Theorem \ref{Theorem2_ForwardUtility}, we  have the following result.
\begin{proposition}
\label{prop-example2}
Suppose that $(Y(t),Z_1(t),Z_2(t),\lambda)=(y(V_t),z^1(V_t),z^2(V_t),\lambda)$ is the Markovian solution of ergodic BSDE \eqref{EQBSDE1} with the driver \eqref{120901}. Then, the process  $U(x,t)$ given by
\begin{equation*}
U(x,t)=\frac{x^{\delta }}{\delta }e^{y(V_{t})-\lambda t}\text{ ,}
\end{equation*}%
is a power robust forward performance process.
Moreover, the optimal portfolio weights and worst-case scenario parameters  are given by
\begin{equation}\nonumber
\begin{aligned}
\pi^{\ast }_1(t)=&\frac{1}{1-\delta}\Big(\theta(V_t)+Z_1(t)\\
&+Proj_{[-R,R]}\big(-\theta(V_t)
-\frac{1}{\delta}Z_1(t)-\frac{1-\delta}{\delta}Z_2(t)I_{\{Z_2(t)\geq 0\}}\big)\Big),\\
 \pi ^{\ast}_2(t)=&0,\\
u^{\ast}_1(t)=&Proj_{[-R,R]}\Big(-\theta(V_t)-\frac{1}{\delta}Z_1(t)-\frac{1-\delta}{\delta}Z_2(t)\cdot I_{\{Z_2(t)\geq 0\}}\Big),\\
u^{\ast}_2(t)= &Proj_{[-R,R]}\Big(-\theta(V_t)-\frac{1}{\delta}Z_1(t)-\frac{1-\delta}{\delta}Z_2(t)\Big)\cdot I_{\{Z_2(t)\geq 0\}}+R\cdot I_{\{Z_2(t)< 0\}}.
\end{aligned}
\end{equation}
 The optimal portfolio weight strategies for each scenario $u\in\mathcal{U}$ and the worst case scenario strategies for each investment weight $\pi\in\tilde{\Pi}$ are given as follows
\begin{equation*}
\begin{aligned}
\alpha ^{\ast }_1(t,u(t)) &=\frac{1}{1-\delta} [\theta(V_t)+Z_1(t)+u_1(t)],\ \alpha^{\ast}_2(t,u(t))=0,\\
\beta^{*}_1(t,\pi(t))&=\left\{
\begin{array}{ll}
u_1^*(t),\ &\text{if}\ \pi_1(t)=\pi^{*}_1(t),\\
-R\cdot sgn(a(t)),\ &\text{otherwise,}
\end{array}
\right.\\
\beta^{*}_2(t,\pi(t))&=\left\{
\begin{array}{ll}
u_2^*(t),\ &\text{if}\ \pi_1(t)=\pi^{*}_1(t),\\
-R\cdot sgn(a(t))\cdot I_{\{Z_2(t)\geq 0\}}+R\cdot I_{\{Z_2(t)<0\}},\ &\text{otherwise,}
\end{array}
\right.
\end{aligned}
\end{equation*}
where $a(t):=\delta\pi_1(t)+Z_1(t)+Z_2(t)\cdot I_{\{Z_2(t)\geq 0\}}$.
\end{proposition}

Similar to Remark \ref{080701}, we can conclude that the best choice
for an investor will take no action in an incomplete market if the
uncertainty is large enough. In fact, from Proposition
\ref{prop-example2} (especially, the form of $\pi_1^*$)  we can
derive a boundary (say, $M$)  or a domain of the scenario's value in
what degree an ambiguity-averse investor should not take action in
the market with uncertainty. Herein, the boundary $M$ depends on the
boundedness of the functions $z^1(\cdot)$ and $z^2(\cdot)$ as well
as $\theta(\cdot)$. Once the uncertainty exceeds the boundary $M$,
the robust investment opportunity will disappear (in this case,
$\pi_1^*\equiv 0$ because $Z_1=Z_2\equiv 0$). In turn, this reflects
that the uncertainty in our model is essentially associated with the
risk price $\theta$ and the part $Z$ of the solution of ergodic
BSDE.

Moreover, when the stock price is not affected by the stochastic
factor in the sense that the coefficients $b$ and $\sigma$ in
\eqref{120903} are constants, the processes $Z_1$ and $Z_2$, as part
of the solution of the ergodic BSDE \eqref{EQBSDE1}, will equal to
0. Then, from Proposition \ref{prop-example2}, it is easy to check
that the worst-case scenario parameters $u^*_1$ and $u^*_2$ will
choose the values closest to $-\theta(=-\frac b\sigma)$ for any
given $R>0$ and the optimal portfolio weight $\pi_1^*$ will have the
form
$$\pi^{\ast }_1(t)=\frac{1}{1-\delta}\Big(\theta+Proj_{[-R,R]}\big(-\theta\big)\Big).$$
Therefore, for the model that the stock price is not affected by the stochastic factor,  there will be no investment action into the stock when the value of the risk $\theta$ belongs to the range of the uncertainty (i.e., $|\theta|\leq R$).

In addition, we observe that the sign of $z^2(V_t)(=Z_2(t))$ has an important impact on the  the worst-case scenario, albeit not shown explicitly in the form of the power robust forward performance process $U(x,t)$.  
It seems interesting to observe that the sign of $z^2(V_t)(=Z_2(t))$
only affects the worst-case scenario strategies $\beta^*_1$ and
$\beta^*_2$, not the optimal investment policy strategies
$\alpha^*_1$ and  $\alpha^*_2$ responding to each scenario. A
similar situation occurs if one consider a general compact and
convex subset $U\subset\mathbb{R}^2$ (e.g. $U=\{(u_1,u_2): -R\leq
u_i \leq R,\ i=1,2\}$); the only difference is that for this case
the form of worst-case scenario parameters depend also on the sign
of some process involving $z^1(V_t)$. Therefore, one may deduce that
the $Z$'s part of the solution of the ergodic BSDE \eqref{EQBSDE1}
carries with the important information on the worst-case scenario.

\begin{remark} The above incomplete market model with
uncertainty has also been studied in \cite{HS2006} in the framework
of classical robust expected utility. They give an explicit PDE
characterization for the lower value function of a robust utility
maximization problem  combining the duality approach and the
stochastic control approach.

On the other hand, when we do not
consider the model uncertainty, the above model will reduce to the
case that has been studied in \cite{LZ} (Section 3.1.3 therein). The
optimal portfolio weights obtained in \cite{LZ} have the following
form
$$\tilde{\pi}_1^*(t)=\frac{1}{1-\delta}(\theta(V_t)+\tilde{Z}_1(t)),\ \tilde{\pi}_2^*(t)=0,$$
where $(\tilde{Y},\tilde{Z}_1,\tilde{Z}_2,\tilde{\lambda})$ is the Markovian solution of ergodic BSDE \eqref{EQBSDE1} with the driver
$$\tilde{G}(v,z_1,z_2)=\frac{1}{2} \frac{\delta}{1-\delta}|z_1+\theta(v)|^2+\frac{1}{2}(|z_1|^2+|z_2|^2).$$
 Comparing the form between optimal portfolio weight $\tilde{\pi}_1^*$ and the robust weight $\pi_1^*$ given in Proposition \ref{prop-example2}, we observe that the model uncertainty affects the optimal policy in the following two aspects:\\
{\rm i)} $\pi^*_1$ has an additional projection term, which can be seen as a direct reflection on the model uncertainty influencing the robust investment policy;\\
{\rm ii)} The  solutions of the ergodic BSDE \eqref{EQBSDE1} with driver
$G$ and $\tilde{G}$, especially for the $Z$'s part shown in
$\pi_1^*$ and $\tilde{\pi}_1^*$, are different. Note that the
difference between $G$ and $\tilde{G}$ is mainly caused by the model
uncertainty. This reflects indirectly the impact on optimal policy
induced by the uncertainty via the associated ergodic BSDE.
\end{remark}
\section{Connection with ergodic risk-sensitive stochastic differential games}
\label{Connection with risk-sensitive optimization}
We establish a connection between the constant $\lambda$ appearing
in the solution of the ergodic BSDE (\ref{EQBSDE1}) and  a zero-sum
risk-sensitive stochastic differential game over the infinite
horizon with ergodic payoff criteria. It turns out the constant
$\lambda$ is the value of the zero-sum risk-sensitive game  and can
be interpreted as the optimal long-term growth rate of expected
utility of wealth with model uncertainty.

We first give the comparison theorem for ergodic BSDE \eqref{EQBSDE1}, which will be employed in Theorem \ref{propositionLambda}. Moreover, this result can be applied to compare the robust optimal long-term growth rate of expected utility for the model with different parameters.
\begin{lemma}\label{le112302}
Suppose that $G_i$, $i=1,2,$ satisfy the following conditions
\begin{equation}\label{112401}
\begin{aligned}
&|G_i(v,z)-G_i(\bar{v},z)|\leq C(1+|z|)\cdot|v-\bar{v}|,\\
&|G_i(v,z)-G_i(v,\bar{z})|\leq C(1+|z|+|\bar{z}|)\cdot|z-\bar{z}|,\
|G_i(v,0)|\leq C.
\end{aligned}
\end{equation}
For $i=1,2,$ let $(Y^i,Z^i,\lambda^i)$ be the unique Markovian solution of the ergodic BSDE \eqref{EQBSDE1}
with driver $G_i(v,z)$. If $G_1(v,z)\geq G_2(v,z)$, then we have
$$\lambda^1\geq\lambda^2.$$
\end{lemma}
We remark that, under the assumptions  on the coefficients of our
model (mainly the boundedness and Lipschitz assumption on the market
price of the risk $\theta$), the function $G(v,z)$ defined in
\eqref{driver-formal} satisfies the condition \eqref{112401} (see
\eqref{120401}). We next give the proof of Lemma \ref{le112302}.

\begin{proof}
Denote
$$\gamma_t=\left\{
\begin{array}{ll}
\frac{G_1(V_t,Z_t^1)-G_1(V_t,Z_t^2)}{|Z^1_t-Z^2_t|^2}(Z_t^1-Z_t^2),\ &\text{if}\ Z_t^1\neq Z_t^2,\\
0,\ & \text{otherwise.}
\end{array}
\right.$$
Then, from the boundedness of $Z^1$ and $Z^2$, we know $\gamma$ is a bounded process.
We define the probability measure $Q$ as follows
$$\frac{dQ}{d\mathbb{P}}\Big|_{\mathcal{F}_t}=\mathcal{E}(\int_0^t\gamma_rdW_r).$$
Using the notations $\hat{Y}=Y^1-Y^2$, $\hat{Z}=Z^1-Z^2$, $\hat{\lambda}=\lambda^1-\lambda^2$, we get
\begin{equation*}
\begin{aligned}
\hat{Y}_0-\hat{Y}_T&=\int_0^TG_1(V_t,Z_t^2)-G_2(V_t,Z_t^2)+\gamma_t^T\hat{Z_t}dt-\hat{\lambda}T-\int_0^T\hat{Z}_t^TdW_t\\
&=\int_0^TG_1(V_t,Z_t^2)-G_2(V_t,Z_t^2)dt-\hat{\lambda}T-\int_0^T\hat{Z}_t^TdW_t^Q,
\end{aligned}
\end{equation*}
where  $W^Q$ defined via $dW^Q_t=-\gamma_tdt+dW_t$ is a Brownian motion under the probability measure $Q$. Therefore, we get
\begin{equation}
\label{112201}
\frac 1T
E_Q[\hat{Y}_0-\hat{Y}_T]+\hat{\lambda}=\frac 1T E_Q[\int_0^TG_1(V_t,Z_t^2)-G_2(V_t,Z_t^2)dt].
\end{equation}
Note that there exist mappings $y^i,$ $i=1,2,$ such that
$Y_t^i=y^i(V_t),$ $i=1,2$. Since $y^i$, $i=1,2,$ are of linear
growth,  there exists a constant $C$ independent of $T$ such that
\begin{equation}\label{112304}
E_Q|\hat{Y}_T|\leq C(1+E_Q|V_T|)\leq C,
\end{equation}
where the last inequality is derived from the dissipative condition \eqref{dissipative}.
It follows from  \eqref{112201} and $G_1(v,z)\geq G_2(v,z)$ that
$$\hat{\lambda}=\limsup_{T\rightarrow\infty}\frac 1T E_Q[\int_0^TG_1(V_t,Z_t^2)-G_2(V_t,Z_t^2)dt]\geq 0,$$
which completes the proof.
\end{proof}

{ We start to formulate a two-player
zero-sum risk-sensitive stochastic differential game  associated
with the forward process. The dynamic is given by the stochastic
factor model \eqref{factor-SDE} and the running payoff function is
given by
\begin{equation*}
L(v,\pi,u)=-\frac{1}{2}\delta (1-\delta )|{\pi }|^{2}+\delta\pi^T[
\theta (v)+u],\ \ (v,\pi,u)\in\mathbb{R}^d\times\mathbb{R}^d\times\mathbb{R}^d.
\end{equation*}
In line with the forward process, the planning horizon of the
differential game is infinite and we study the following ergodic
payoff criterion
\begin{equation}\label{2020090401}
\mathcal{J}(\pi,u)=\limsup_{T\uparrow \infty }\frac{1}{T}%
\ln E_{{\mathbb{P}}^{{\pi,u }}}\left( e^{\int_{0}^{T}L(V_{s},{\pi }%
_{s},u_s)ds}\right),\ \ (\pi,u)\in\tilde{\Pi}\times\mathcal{U},
\end{equation}
where the probability measure  ${\mathbb{P}}^{{\pi,u}}$ is defined as follows
\begin{equation} \label{RN_density_21}
\left. \frac{d{\mathbb{P}}^{{\pi,u }}}{d\mathbb{P}}\right
\vert _{\mathcal{F}_{t}}=\mathcal{E}\left( \int_{0}^{t }(\delta {\pi }%
_{r}^{T}+u_r^T)dW_{r}\right) .
\end{equation}%
Note that the criterion $\mathcal{J}$ represents the gain for Player 1 and the loss for Player 2. Thus, Player 1 aims to maximize  $\mathcal{J}$ by using her control $\pi$, whereas Player 2 wants to minimize it  via her control $u$. Intuitively, this model can be applied to describe the long-time investment action of a risk-averse investor in a market with model uncertainty (see Remark \ref{20200921} on the equivalent form of $\mathcal{J}$), namely, the investor is trying to maximize her long-term portfolio gain rate via choosing the portfolio weight $\pi$, whereas the market, by default, aims to minimize the investor's gain rate via hiding the real market model and adding disturbance terms.
}

\begin{theorem}
\label{propositionLambda} For any
$(\pi,u)\in\tilde{\Pi}\times\mathcal{U}$ with feedback forms, i.e.
$(\pi_s,u_s)=(\pi(V_s),u(V_s))$ for some Borel measurable mappings
$(\pi(\cdot),u(\cdot))$, let $(y(V_{t}),$ $z\left( V_{t}\right) ,$
$\lambda ),$ $t\geq 0,$ be the unique Markovian solution of the
ergodic BSDE (\ref{EQBSDE1}).
 Furthermore, if the set $\Pi$ is also assumed to be bounded, then $\lambda $ is the
 value of the associated risk-sensitive game problem, namely,
\begin{equation}\label{ErgodicGameProblem1}
\begin{aligned}
\lambda =\inf_{u\in\mathcal{U}}\sup_{{\pi }\in \tilde{\Pi}} \mathcal{J}(\pi,u)
 =\sup_{{\pi }\in \tilde{\Pi}}\inf_{u\in\mathcal{U}}\mathcal{J}(\pi,u).
\end{aligned}
\end{equation}%
 Moreover, the supremum  and infimum in \eqref{ErgodicGameProblem1} can be attainable by choosing  $\pi^*$ and $u^*$  as in \eqref{112301}.
\end{theorem}

\begin{proof}
From \eqref{112001}, we have
\begin{equation}\label{112501}
\begin{aligned}
&|F(v,z,\pi,u)-F(\bar{v},z,\pi,u)|\leq C|\pi|\cdot|v-\bar{v}|,\\
&|F(v,z,\pi,u)-F(v,\bar{z},\pi,u)|\leq C(1+|\pi|+|z|+|\bar{z}|)\cdot|z-\bar{v}|,\\
&|F(v,0,\pi,u)|\leq C|\pi|^2+C|\pi|.\\
\end{aligned}
\end{equation}
Then, similar to the proof of Lemma \ref{le112301}, from
\eqref{112302} and \eqref{112501}, the following two ergodic
equations
\begin{equation}\label{EQBSDE2}
\begin{aligned}
dY_{t}^u&=(-\sup_{\pi_t\in {\Pi}}F(V_{t},Z_{t}^u,\pi_t,u_t)+\lambda^u )dt+(Z_{t}^u)^{T}dW_{t},  \\
dY_{t}^\pi&=(-\inf_{u_t\in {U}}F(V_{t},Z_{t}^\pi,\pi_t,u_t)+\lambda^\pi )dt+(Z_{t}^\pi)^{T}dW_{t},
\end{aligned}
\end{equation}%
have unique Markovian solutions $(Y^u,Z^u,\lambda^u)$ and $(Y^\pi,Z^\pi,\lambda^\pi)$, respectively, for each $(u,\pi)\in\mathcal{U}\times\tilde{\Pi}$ with feedback forms.

\emph{Step 1.} We first show that
\begin{equation}\label{lambda}
\lambda =\inf_{u\in\mathcal{U}}\lambda^u  =\sup_{{\pi }\in \tilde{\Pi}}\lambda^\pi.
\end{equation}
Since
$$\inf_{u_t\in {U}}F(V_{t},Z_{t},\pi_t,u_t)\leq G(V_t,Z_t)\leq \sup_{\pi_t\in {\Pi}}F(V_{t},Z_{t},\pi_t,u_t),$$
Lemma \ref{le112302} then implies that
\begin{equation}\label{112406}
\lambda^\pi\leq\lambda\leq\lambda^u,\ \text{for\ all}\ (\pi,u)\in\tilde{\Pi}\times\mathcal{U}\ \text{with feedback forms}.
\end{equation}
On the other hand, from the uniqueness of the solution of the
ergodic BSDE (\ref{EQBSDE1}), we know
$\lambda=\lambda^{u^*}=\lambda^{\pi^*}$ with $u^*$ and $\pi^*$ given
in \eqref{112301}. Thus, we have established (\ref{lambda}).

\emph{Step 2.} We show that, for each $(u,\pi)\in\mathcal{U}\times\tilde{\Pi}$ with feedback forms,
\begin{equation}\label{112407}
\lambda^u=\sup_{{\pi }\in \tilde{\Pi}}\limsup_{T\uparrow \infty }\frac{1}{T}%
\ln E_{{\mathbb{P}}^{{\pi,u }}}\left( e^{\int_{0}^{T}L(V_{s},{\pi }%
_{s},u_s)ds}\right),
\end{equation}
\begin{equation}\label{112408}
\lambda^\pi=\inf_{{u }\in \mathcal{U}}\limsup_{T\uparrow \infty }\frac{1}{T}%
\ln E_{{\mathbb{P}}^{{\pi,u }}}\left( e^{\int_{0}^{T}L(V_{s},{\pi }%
_{s},u_s)ds}\right).
\end{equation}
We only prove \eqref{112407}, and the proof of \eqref{112408} is
analogous.

For arbitrary but fixed $u\in\mathcal{U}$, from \eqref{EQBSDE2} we get, for every $\tilde{\pi}\in\tilde{\Pi}$,
\begin{equation}\label{112303}
\begin{aligned}
dY_{t}^u&=(-\sup_{\pi_t\in {\Pi}}F(V_{t},Z_{t}^u,\pi_t,u_t)+\lambda^u )dt+(Z_{t}^u)^{T}dW_{t}\\
&=\Big(-\sup_{\pi_t\in {\Pi}}F(V_{t},Z_{t}^u,\pi_t,u_t)+\lambda^u +(Z_t^u)^T(\delta\tilde{\pi}_t+u_t)\Big)dt+(Z_{t}^u)^{T}dW_{t}^{\tilde{\pi},u},
\end{aligned}
\end{equation}%
where $W^{\tilde{\pi},u}$ defined via $dW^{\tilde{\pi},u}=-(\delta\tilde{\pi}_t+u_t)dt+dW_t$ is a Brownian motion under probability measure $\mathbb{P}^{\tilde{\pi},u}$ (see \eqref{RN_density_21}).
We observe that the function $F$ (see \eqref{112001}) in \eqref{112303}
can be written as
\begin{equation*}
F(V_{t},Z_{t}^u,\pi_t,u_t)=L(V_{t},\pi
_{t},u_t)+(Z_{t}^u)^{T}(\delta {\pi }_{t}+u_t)+\frac{1}{2}|Z_{t}^u|^{2}.
\end{equation*}%
Therefore,  we rewrite the
ergodic BSDE (\ref{112303}) as
\begin{equation*}
\begin{aligned}
&Y_{0}^u-Y_T^u+\lambda^u T\\
=&\int_0^T\sup_{{\pi }_{t}\in \tilde{\Pi} }\left( L(V_{t},\pi
_{t},u_t)+(Z_{t}^u)^{T}\delta {\pi }_{t}\right) -(Z_{t}^u)^{T}\delta \tilde{\pi}%
_{t}+\frac{1}{2}|Z_{t}^u|^{2}dt-\int_0^T(Z_{t}^u)^{T}dW_{t}^{
\tilde{\pi},u},
\end{aligned}
\end{equation*}%
which follows that, for arbitrary $\tilde{\pi}\in\tilde{\Pi}$,
\begin{equation}\nonumber
\begin{aligned}
&e^{\lambda^uT+Y_{0}^u}e^{-Y_{T}^u}\mathcal{E}\Big( \int_{0}^{T
}(Z_{t}^u)^{T}dW_{t}^{\tilde{\pi},u}\Big) \\
=& \exp \Big( \int_0^T\sup_{{\pi }_{t}\in \tilde{\Pi} }\left( L(V_{t},\pi
_{t},u_t)+(Z_{t}^u)^{T}\delta {\pi }_{t}\right) -L(V_{t},\tilde{\pi}
_{t},u_t)-(Z_{t}^u)^{T}\delta \tilde{\pi}%
_{t}dt\Big)\\
&\cdot e^{\int_{0}^{T}L(V_{t},\tilde{\pi}_{t},u_t)dt}\\
\geq &\ e^{\int_{0}^{T}L(V_{t},\tilde{\pi}_{t},u_t)dt}.
\end{aligned}
\end{equation}%
Then, we obtain
\begin{equation}\label{112306}
\begin{aligned}
e^{\lambda^uT+Y_{0}^u}E_{\mathbb{P}^{\tilde{\pi},u}}\Big[e^{-Y_{T}^u}\mathcal{E}\Big( \int_{0}^{T
}(Z_{t}^u)^{T}dW_{t}^{\tilde{\pi},u}\Big) \Big]
\geq E_{\mathbb{P}^{\tilde{\pi},u}}\Big[e^{\int_{0}^{T}L(V_{t},\tilde{\pi}_{t},u_t)dt}\Big].
\end{aligned}
\end{equation}%
We define the probability measure $Q^{\tilde{\pi},u}$ as follows
\begin{equation}\nonumber
\left. \frac{d{Q}^{{\tilde{\pi},u }}}{d\mathbb{P}}\right
\vert _{\mathcal{F}_{t}}=\mathcal{E}\left( \int_{0}^{t }(\delta {\tilde{\pi} }%
_{r}+u_r+Z_r^u)^TdW_{r}\right) .  \label{RN_density_2}
\end{equation}%
Using the measure $Q^{\tilde{\pi},u}$, from \eqref{112306} we get
\begin{equation}\nonumber
\begin{aligned}
e^{\lambda^uT+Y_{0}^u}E_{Q^{\tilde{\pi},u}}\Big[e^{-Y_{T}^u} \Big]
\geq E_{\mathbb{P}^{\tilde{\pi},u}}\Big[e^{\int_{0}^{T}L(V_{t},\tilde{\pi}_{t},u_t)dt}\Big].
\end{aligned}
\end{equation}%
Thus, it holds
\begin{equation}\label{112307}
\begin{aligned}
\lambda^u+\frac{Y_{0}^u}{T}+\frac 1T \ln
E_{Q^{\tilde{\pi},u}}\left[e^{-Y_{T}^u} \right] \geq \frac 1T \ln
E_{\mathbb{P}^{\tilde{\pi},u}}\Big[e^{\int_{0}^{T}L(V_{t},\tilde{\pi}_{t},u_t)dt}\Big].
\end{aligned}
\end{equation}%
Similar to the proof of estimate \eqref{112304}, from the boundedness of $\Pi$ and  Jensen's inequality, there exists a constant $C$
independent of $T$ such that
\begin{equation}\label{112308}
\frac{1}{C}\leq  e^{-E_{Q^{\tilde{\pi },u}}\left[Y_{T}^u\right]}\leq
E_{Q^{\tilde{\pi },u}}\left( e^{-Y_{T}^u}\right) \leq C,
\end{equation}
where the last inequality is obtained using Lemma 3.1 in \cite{FMc}.
It follows from \eqref{112307} and \eqref{112308} that, for any $\tilde{\pi}\in
\tilde{\Pi}$,
\begin{equation*}
\lambda^u \geq \limsup_{T\uparrow \infty }\frac 1T \ln E_{\mathbb{P}^{\tilde{\pi},u}}\Big[e^{\int_{0}^{T}L(V_{t},\tilde{\pi}_{t},u_t)dt}\Big].
\end{equation*}%
with equality choosing $\tilde{\pi}_{t}=\pi _{t}^{*}$, where $\pi
_{t}^{*}$ is given in (\ref{112301}).

\emph{Step 3.} Finally, we readily obtain  \eqref{ErgodicGameProblem1} from (\ref{lambda}) in Step 1 and (\ref{112407}) and (\ref{112408}) in Step 2.
\end{proof}

{\begin{remark}\label{20200921} Note
that
\begin{equation}\nonumber
\begin{aligned}
&E_{{\mathbb{P}}^{{\pi,u }}}\left( e^{\int_{0}^{T}L(V_{s},{\pi }%
_{s},u_s)ds}\right)\\
=&\ E_{{\mathbb{P}}^{{u }}}\left( e^{\int_{0}^{T}-\frac{1}{2}\delta
|{\pi }_{s}|^{2}+\delta\pi_s^T \theta
(V_{s})ds+\int_0^T\delta\pi_s^TdW_s}\right) = E_{{\mathbb{P}}^{{u
}}}\left[\frac{(X_T^\pi)^\delta}{\delta}\right]\cdot
\frac{\delta}{x^\delta}.
\end{aligned}
\end{equation}
In turn, from \eqref{ErgodicGameProblem1}, it is easy to check that
$\lambda$ is also the value for the following game problem
\begin{equation}\nonumber
\begin{aligned}
\lambda &=\inf_{u\in\mathcal{U}}\sup_{{\pi }\in \tilde{\Pi}}\limsup_{T\uparrow \infty }\frac{1}{T}%
\ln E_{{\mathbb{P}}^{u }} \left[\frac{(X_T^\pi)^\delta}{\delta}\right] \\
& =\sup_{{\pi }\in \tilde{\Pi}}\inf_{u\in\mathcal{U}}\limsup_{T\uparrow \infty }\frac{1}{T}%
\ln E_{{\mathbb{P}}^{u }} \left[\frac{(X_T^\pi)^\delta}{\delta}\right].
\end{aligned}
\end{equation}%
Therefore, Theorem \ref{propositionLambda} can be viewed as an
optimal investment model, and the constant $\lambda$ is the optimal
long-term growth rate of the expected utility of wealth with model
uncertainty. Such asymptotic results on robust utility maximization
have been treated in \cite{Knispel} using the duality method and are
related to ``robust large deviations" criteria to optimal long-term
investment, that is, the investor aims to maximize the portfolio's
growth rate exceeding some threshold $C\in\mathbb{R}$ under the
worst-case probability
\begin{equation}\nonumber
\begin{aligned}
\sup_{{\pi }\in \tilde{\Pi}}\inf_{u\in\mathcal{U}}\limsup_{T\uparrow \infty }\frac{1}{T}%
\ln {\mathbb{P}}^{u }\big(\frac 1 T \ln {X_T^\pi}\geq C\big).
\end{aligned}
\end{equation}%
\end{remark}}

\section{Connection with classical expected utility maximization for long time
horizons}

\label{SubSectionTraditionalPowerUtility}

We establish a link between  the power robust forward process $U\left( x,t\right) $
 and  the long-time behaviour of the lower value function  of the classical power robust expected utility.
For the latter, let $[0,T]$ be an arbitrary trading horizon and we introduce the lower value function as follows
\begin{equation}
w_T(x,v)=\sup_{\pi \in {\Pi_{[0,T]}}}\inf_{u\in\mathcal{U}_{[0,T]}}E_{\mathbb{P}^u}\left[
\frac{(X_{T}^{\pi })^{\delta }}{\delta }|X_0^\pi=x,V_0=v\right], \  (x,v) \in \mathbb{R}_{+}\times\mathbb{R}^d, \label{PowerUtility}
\end{equation}%
where the wealth process $X_{s}^{\pi }$, $s\in \lbrack 0,T],$ solving (\ref%
{wealth-process}) with $X_0^\pi=x$, the stochastic factor process
$V_s$, $s\in[0,T]$, solving (\ref{factor-SDE}) with $V_0=v$, and
$u\in\mathcal{U}_{[0,T]}$ means that $u$ belongs to $\mathcal{U}$
and is restricted to the time horizon $[0,T]$.

We recall that the optimal investment problem for the classical robust expected utility  has been considered in \cite{BMS} via the stochastic control approach  based on BSDE, in  \cite{S2007} via the duality approach, and in \cite{HS2006} combining these two methods.

\begin{proposition}
\label{PropositionPowerUtility2}
Let $U\left( x,t\right)=\frac{x^{\delta}}{\delta}e^{y(V_t)-\lambda t} $ be the power robust forward performance process as in (\ref%
{PowerForwardUtility}). Then, there exists a constant  $L\in\mathbb{R}$, independent of the initial states $X_0^\pi=x$ and $V_0=v$,
such that, for $ (x,v)\in \mathbb{R}%
_{+}\times\mathbb{R}^d,$
\begin{equation*}
\lim_{T\uparrow \infty }\frac{w_T(x,v)e^{-\lambda T-L} }{U(x,0)}=1.
\end{equation*}%
\end{proposition}

\begin{proof}
Since the maxmin problem (\ref{PowerUtility}) is standard in the literature (see, for example, \cite{YLZ}), we only demonstrate its main steps briefly. To this end,  for each $\pi \in {\Pi_{[0,T]}}$ and $u\in \mathcal{U}_{[0,T]}$, we introduce the objective functional
$$w_T(x,v,\pi,u)=E_{\mathbb{P}^u}\left[\frac{(X_T^{\pi})^{\delta}}{\delta}|X_0^\pi=x,V_0=v\right].$$
We aim to find a saddle point $(\pi^*,u^*)\in\Pi_{[0,T]}\times \mathcal{U}_{[0,T]}$ such that
$$w_T(x,v,\pi,u^*)\leq w_T(x,v,\pi^*,u^*)\leq w_T(x,v,\pi^*,u).$$
Then, it is clear that $w_T(x,v)=w_T(x,v,\pi^*,u^*)$. We claim that
\begin{equation}\label{value_function}
w_T(x,v)=\frac{x^\delta}{\delta}e^{\bar{Y}_0},
\end{equation}
\begin{equation}\label{optimal_pi_u}
\pi_t^*=\pi^*(V_t,\bar{Z}_t),\ u_t^*=u^*(V_t,\bar{Z}_t),\ t\in[0,T],
\end{equation}
with the mappings $(\pi^*,u^*)$ given in \eqref{optimal-3} and  \eqref{optimal-2}, and $(\bar{Y},\bar{Z})$ being the unique solution of the following BSDE
\begin{equation}\label{112504}
\bar{Y}_t=\int_{t}^{T}G(V_{r},\bar{Z}_{r})dr-\int_{t}^{T}\left(
\bar{Z}_{r}\right) ^{T}dW_{r},
\end{equation}%
where the driver $G$ is given in \eqref{driver-formal}. The proof
follows along similar arguments as in Proposition \ref{le-PDE} and
Theorem \ref{Theorem2_ForwardUtility}, and thus omitted.

From Theorem 4.4 in \cite{Hu11}, there exists a constant $L\in\mathbb{R}$ such that
\begin{equation}\label{112509}
\lim_{T\uparrow\infty}(\bar{Y}_0-\lambda T-Y_0)=L,
\end{equation}
where $(Y,Z,\lambda)$ is the solution of ergodic BSDE \eqref{EQBSDE1}. Finally,
from \eqref{PowerForwardUtility}, (\ref{value_function}) and \eqref{112509} we have
\begin{equation*}
\lim_{T\uparrow \infty }\frac{w_T(x,v)e^{-\lambda T-L} }{U(x,0)}=1,
\end{equation*}
which completes the proof.%
\end{proof}

{We have showed that the discounted
classical power robust expected utility converges to the power
robust forward performance process as the investment horizon tends
to infinite. This result is obtained via the relationship between
BSDE and ergodic BSDE. Then it seems natural and interesting to
consider the connection of the optimal strategies between these two
investment problems. In fact, from \eqref{112301} we get that the
optimal investment strategy to each scenario
$$\alpha^*(t,u_t)=Proj_{\Pi}(\frac{\theta(V_t)+Z_t+u_t}{1-\delta}),$$
where $(Y,Z,\lambda)$ is the solution of ergodic BSDE \eqref{EQBSDE1}. Similar arguments conclude that for classical expected utility problem (\ref{PowerUtility}), the optimal investment strategy has the form
$$\alpha^*_T(t,u_t)=Proj_{\Pi}(\frac{\theta(V_t)+\bar{Z}_t+u_t}{1-\delta}),$$
where $(\bar{Y},\bar{Z})$ is the unique solution of BSDE
\eqref{112504}. Since the projection operator on a closed convex set
is Lipschitz continuous, the convergence of $\alpha^*_T$ to
$\alpha^*$ boils down to the convergence of
\begin{equation}\label{convergence_Z}
E_{\mathbb{P}}[\int_0^T|Z_t-\bar{Z}_t|^2dt]\rightarrow 0,\
\text{as}\ T\rightarrow\infty.
\end{equation} However, (\ref{convergence_Z}) is
yet to be established. Hence, although we have established a
connection of the robust forward performance process with the
corresponding classical robust expected utility, it is still an open
problem to prove the convergence of the associated optimal trading
strategies.}

\section{{Logarithmic robust forward performance processes with negative realization processes}}

\label{SDGsframework}

In this section, we provide an example to show the advantage of our
stochastic differential approach to solve homothetic robust forward
investment problems compared with the classical saddle point
argument.
For simplicity of the calculations, we consider the following
situation:

i) a single stock and single stochastic factor model (i.e., $n=d=1$
in state equations \eqref{stock-SDE} and\eqref{factor-SDE}).

ii) $U=[0,1]$, $\Pi=[0,1]$, the market price of the risk $\theta\in
[-1,0]$.

iii) a logarithmic robust forward performance process,
$$U(x,t)=\ln x+f(V_t,t)$$
with a quadratic form on the {realization process $\gamma$ in
\eqref{maxmim}, i.e.,
\begin{equation}\label{2020082201}
\gamma_{t,s}(u)=\int_t^s\frac{1}{2}|u_s|^2ds,
\end{equation}
which expresses the cumulative evaluation of the model $\mathbb{P}^u$ predicted by the investor in $[t,s]$.
Herein, since we consider negative realization processes, we choose the parameter $\tau=-1$ without loss of generality.} 
Using similar arguments to Theorem \ref{Theorem2_ForwardUtility}, we
obtain the following results.
\begin{theorem}
\label{ForwardUtilitywithrewardfunction}
Let $(\widetilde{Y}_{t},\widetilde{Z}_{t},\widetilde{\lambda} )=(\widetilde{y}(V_{t}),\widetilde{z}(V_{t}),\widetilde{\lambda} ),t\geq 0,$ be the
unique Markovian solution of (\ref{EQBSDE1}) with the generator
$$\widetilde{G}(v,z)=\max_{\pi\in\Pi}\min_{u\in U}\widetilde{F}(v,z,\pi,u),$$
where $\widetilde{F}(v,z,\pi,u)=-\frac{1}{2}\pi^2+\pi\theta(v)+(\pi+z)u-\frac{1}{2}u^2.$
Then,
 the process $U(x,t),$ $\left( x,t\right) \in \mathbb{R}_{+}\times \left[
0,\infty \right) ,$ given by
\begin{equation}\nonumber
U(x,t)=\ln x+ \widetilde{y}(V_{t})-\widetilde{\lambda} t,
\end{equation}%
is a logarithmic robust forward performance process with realization process
 $\gamma$ given in \eqref{2020082201} and parameter
$\tau=-1$.
Moreover, the optimal portfolio weight $\widetilde{\pi}^*$,  and
the worst-case scenario strategy $\widetilde{\beta}^*$ responding to each portfolio weight $\pi$ are given as follows
\begin{equation}\nonumber
\widetilde{\pi}^*_t=\widetilde{\pi}^*(V_t,\widetilde{z}(V_t)),\
\widetilde{\beta}^*(t,\pi_t)=\widetilde{\beta}^*(\widetilde{z}(V_t),\pi_t),
\end{equation}
where the mappings $(\widetilde{\pi}^*,\widetilde{\beta}^{*})$ have the form
$$\widetilde{\pi}^*(v,z)=\left\{
\begin{array}{ll}
\theta(v)+1, & \text{if}\ \ \frac 12-z\geq\theta(v)+1;\\
\frac 12-z, & \text{if}\ \ \theta(v)\leq \frac 12-z\leq \theta(v)+1;\\
\theta(v), & \text{if}\ \ \frac 12-z\leq\theta(v),
\end{array}
\right.$$
and
$$\widetilde{\beta}^*(z,\pi)=\left\{
\begin{array}{ll}
1, & \text{if}\ \ \pi+z\leq \frac 12;\\
0, & \text{otherwise}.
\end{array}
\right.$$


\end{theorem}

It is easy to check that the saddle point for this forward
investment problem does not exist since $\widetilde{F}$ is concave
both in variables $\pi$ and $u$. Therefore, the classical saddle
point argument can not be applied directly, whereas our stochastic
differential game approach provides an alternative and efficient way
to address this problem.

\section{{Conclusion}}

This paper provides a stochastic differential game framework to
construct a class of forward performance processes with model
uncertainty. In particular, the homothetic robust forward process in
factor form is represented in terms of the Markovian solution of
ergodic BSDE. The approach and results may be extended in several
directions. First, one may consider a general realization process $\gamma$ with
parameter $\tau$ not necessarily zero. Second, it would be
interesting to prove the convergence of the finite time horizon
optimal investment strategy to its forward counterpart as time
horizon becomes large. Both are left for the future research.

\small
\newif \ifabfull \abfulltrue
\providecommand{\bysame}{\leavevmode \hbox to3em{\hrulefill}\thinspace} %
\providecommand{\MR}{\relax \ifhmode \unskip \space \fi MR }
\providecommand{\MRhref}[2]{
\href{http://www.ams.org/mathscinet-getitem?mr=#1}{#2} } \providecommand{%
\href}[2]{#2}


\begin{thebibliography}{99}
\bibitem{AM2019} {{\sc Y. Ait-Sahalia
and F. Matthys,} {\em Robust consumption and portfolio policies when
asset prices can jump}, {Journal of Economic Theory}, 179 (2019),
pp. 1--56.}

\bibitem{bahman} {\sc B. Angoshtari, T. Zariphopoulou, and X.Y. Zhou,}
{\em Predictable forward performance processes: the binomial case}, {SIAM J. Control Optim.}, 58(1) (2020), pp. 327--347.


\bibitem{Michalis}{{\sc M. Anthropelos, T. Geng, and T.
Zariphopoulou,} {\em Competition in fund management and forward
relative performance criteria}, Working paper.
https://arxiv.org/abs/2011.00838.}


\bibitem{ASS} {{\sc L. Avanesyan, M. Shkolnikov, and R. Sircar}, {\em Construction of a
class of forward performance processes in stochastic factor models,
and an extension of Widder's theorem}, {Finance Stoch.}, 24 (2020),} pp. 981--1011.

\bibitem{B1999} {\sc T. Basar}, {\em Nash equilibria of risk-sensitive nonlinear stochastic differential games}, {J. Optim. Theory Appl.}, {100(3) (1999)}, pp. 479--498.

\bibitem{BCZ2018} {{\sc L. Bo, A. Capponi, and C. Zhou}, {\em Power forward performance in semimartingale markets with stochastic integrated factors}, Working paper. https://arxiv.org/abs/1811.11899v2.}

\bibitem{BG2012} {\sc A. Basu and M.K. Ghosh}, {\em Zero-sum risk-sensitive  stochastic differential games}, {Math. Oper. Res.}, {37(3) (2012)}, pp. 437--449.




\bibitem{Bielecki} {\sc T. Bielecki and S. Pliska}, {\em Risk-sensitive
dynamic asset management}, {Appl. Math. Optim.},
39(3) (1999), pp. 337--360.

\bibitem{BS2018} {\sc A. Biswas and S. Saha}, {\em Zero-sum stochastic differential games with risk-sensitive cost}, {Appl. Math. Optim.}, 81(1) (2020), pp. 113--140.

 \bibitem{BMS} {\sc G. Bordigoni, A. Matoussi, and M. Schweizer}, {\em A stochastic control approach to a robust utility maximization problem}, {Stochastic Analysis and Applications, F. E. Benth et al eds., Springer},
 (2007),  pp. 125--151.




\bibitem{BL} {\sc R. Buckdahn and J. Li}, {\em Stochastic differential games
and viscosity solutions of Hamilton-Jacobi-Bellman-Isaacs Equations},
{SIAM J. Control Optim.}, 47(1) (2008), pp. 444--475.

\bibitem{CHLZ} {\sc W. F. Chong, Y. Hu, G. Liang and T. Zariphopoulou,
{\em An ergodic BSDE approach to forward entropic risk measures:
representation and large-maturity behavior}, {Finance and Stoch.},
23(1) (2019), pp. 239--273.}

\bibitem{CL} {\sc W. F. Chong and G. Liang},
{\em Optimal investment and consumption with forward preferences and uncertain parameters}, Working paper. https://arxiv.org/abs/1807.01186.

\bibitem{DCH2018} {{\sc K.-W. Ding, Z.-Y. Chen, and N.-J. Huang}, {\em Robust mean variance optimization problem under R\'enyi divergence information}, {Optimization}, 67(2) (2018), pp. 287--307.}




%

\bibitem{KH2003} {\sc N. El Karoui and S.
Hamad\`ene}, {\em BSDEs and risk-sensitive control, zero-sum and nonzero-sum game problems of stochastic functional differential equations}, {Stoch. Process. Appl.},  107(1) (2003),
pp. 145--169.



\bibitem{FMc} {\sc  W. H. Fleming and W. M. McEneaney}, {\em
Risk-sensitive control on an infinite time horizon}, {SIAM J.
Control Optim.}, 33(6) (1995), pp. 1881--1915.

\bibitem{FS2} {\sc W. H. Fleming and S. J. Sheu}, {\em Risk-sensitive
control and an optimal investment model}, {Math. Finance},
10(2) (2000), pp. 197--213.

\bibitem{FS3} {\sc W. H. Fleming and S. J. Sheu}, {\em Risk-sensitive
control and an optimal investment model {II}}, {Ann. Appl. Probab.},
12(2) (2002), pp. 730--767.

\bibitem{FS1989} {\sc W. H. Fleming and P.E. Souganidis}, {\em
On the existence of value functions of two-player, zero-sum stochastic
differential games}, {Indiana Univ. Math. J.},
38(2) (1989), pp. 293--314.

\bibitem{FG} {\sc H. F\"ollmer and A. Gundel}, {\em Robust projections in the class of martingale measures}, {Illinios Journal of Mathematics}, 50(2)
 (2006), pp. 439--472.



 \bibitem{FSW} {\sc H. F\"ollmer, A. Schied, and S. Weber}, {\em Robust
 preferences and robust portfolio choice},
 {Mathematical Modelling and Numerical Method
 in Finance, P. Ciarlet, et al eds., Handbook of Numerical Analysis, Elsevier},
 15 (2009), pp. 29--88.

\bibitem{HU2} {\sc M. Fuhrman, Y. Hu, and G. Tessitore}, {\em Ergodic
BSDEs and optimal ergodic control in Banach spaces}, {SIAM J.
Control Optim.}, 48 (2009), pp. 1542--1566.



\bibitem{HeStrub} {\sc X. D. He, M. Strub, and T. Zariphopoulou},
{\em Forward rank-dependent performance criteria: Time-consistent investment under probability distortion}, Working paper. https://arxiv.org/abs/1904.01745.

\bibitem{Henderson-Hobson} {\sc V. Henderson and D. Hobson},
{\em Horizon-unbiased utility functions}, {Stochastic Process. Appl.}, 117(11) (2007), pp. 1621--1641.

\bibitem{Henderson3} {\sc V. Henderson and G. Liang}, {\em Pseudo
linear pricing rule for utility indifference valuation}, {Finance
and Stoch.}, 18(3) (2014), pp. 593--615.

\bibitem{HS2006} {\sc D. Hern\'andez-Hern\'andez and A. Schied}, {\em Robust
utility maximization in a stochastic factor model}, {
 Stochastics and Decisions}, 24(1) (2006), pp. 109--125.

\bibitem{HS2007} {\sc D. Hern\'andez-Hern\'andez and A. Schied}, {\em A control approach to robust utility maximization with logarithmic utility and time-consistent penalties}, {Stochastic Process. Appl.}, 117(8) (2007), pp. 980--1000.


\bibitem{HLT}{\sc Y. Hu, G. Liang, and S. Tang},
{\em Systems of ergodic BSDE arising in regime switching forward
performance processes}, {SIAM J. Control Optim.}, 58(4) (2020), pp.
2503--2534.


\bibitem{Hu11} {\sc Y. Hu, P. Madec, and A. Richou}, {\em A
probabilistic approach to large time behaviour of mild solutions of
HJB equations in infinite dimension}, {SIAM J. Control Optim.},
53(1) (2015), pp. 378--398.

\bibitem{KKF2014} {{\sc J. H. Kim, W. C. Kim, and F. J. Fabozzi}, {\em Recent developments in robust portfolios with a worst-case approach}, {J. Optim. Theory Appl.},
161(1) (2014), pp. 103--121.}

\bibitem{KKF2017} {{\sc J. H. Kim, W. C. Kim, and F. J. Fabozzi}, {\em Robust factor-based investing}, {The Journal of Portfolio Management},
43(5) (2017), pp. 157--164.}

\bibitem{KKF2018} {{\sc J. H. Kim, W. C. Kim, and F. J. Fabozzi}, {\em Recent advancements in robust optimization for investment management}, {Annals of Operations Research},
266(1-2) (2018), pp. 183--198. }

\bibitem{Sigrid} {{\sc S. K\"allblad,}
{\em Black's inverse investment problem and forward criteria with
consumption}, {SIAM J. Finan. Math.} 11(2) (2020), pp. 494--525.}

\bibitem{Jan} {\sc S. K\"allblad, J. Ob{\l }{\'o}j, and T.
Zariphopoulou}, {\em Dynamically consistent investment under model
uncertainty: the robust forward criteria}, {Finance Stoch.},  22(4) (2018),
pp. 879--918.


\bibitem{Knispel}
{\sc T. Knispel},  {\em Asymptotics of robust utility maximization},
{Ann. Appl.
Probab.}, 22(1), (2012), pp.172--212.

\bibitem{KS}{\sc N. N. Krasovskii and A. I. Subbotin}, {\em Game-theoretical control problems},
{Springer-Verlag, New York,} (1988).

\bibitem{LSZ} {\sc T. Leung, R. Sircar, and T. Zariphopoulou},
{\em Forward indifference valuation of American options},
{Stochastics}, 84(5-6) (2012), pp. 741--770.

\bibitem{LZ} {\sc G. Liang and T. Zariphopoulou},
{\em Representation of homothetic forward performance processes in stochastic factor models via ergodic and infinite horizon BSDE},
{SIAM J. Financial Math.}, 8 (2017), pp. 344--372.

\bibitem{LSZ2020} {{\sc Q. Lin, X. Sun, and C. Zhou},
{\em Horizon-unbiased investment with ambiguity}, {Journal of
Economic Dynamics and Control}, 114 (2020), 103896.}







\bibitem{Lo} {{\sc A. Lo and M. T. Mueller}, {\em WARNING: Physics envy may be hazardous to your
wealth!}, Working paper. https://arxiv.org/abs/arXiv:1003.2688.}
\bibitem{MZ0} {\sc M. Musiela and T. Zariphopoulou}, {\em Investment
and valuation under backward and forward dynamic exponential
utilities in a stochastic factor model}, {Advances in Mathematical
Finance}, (2007), pp. 303--334.

\bibitem{MZ-Kurtz} {\sc M. Musiela and T. Zariphopoulou}, {\em Optimal
asset allocation under forward exponential performance criteria}, {%
Markov Processes and Related Topics: A Festschrift for T. G. Kurtz,
Lecture Notes-Monograph Series, Institute for Mathematical
Statistics}, 4 (2008), pp. 285--300.


\bibitem{MZ1} {\sc M. Musiela and T. Zariphopoulou}, {\em Portfolio
choice under dynamic investment performance criteria}, {Quant. Finance}, 9(2) (2009), pp. 161--170.

\bibitem{MZ2} {\sc M. Musiela and T. Zariphopoulou}, {\em Portfolio
choice under space-time monotone performance criteria}, {SIAM J.
Financ Math.}, 1 (2010), pp. 326--365.

\bibitem{MZ3} {\sc M. Musiela and T. Zariphopoulou}, {\em Stochastic
partial differential equations and portfolio choice}, {Contemporary
Quantitative Finance}, C. Chiarella and A. Novikov eds., Springer,
Berlin, (2010), pp. 195--215.

\bibitem{NT2017} {\sc S. Nadtochiy and M. Tehranchi}, {\em Optimal
investment for all time horizons and Martin boundary of space-time
diffusions}, {Math. Finance}, 27(2) (2017), pp. 438--470.

\bibitem{NZ} {\sc S. Nadtochiy and T. Zariphopoulou}, {\em A class of
homothetic forward investment performance processes with non-zero
volatility}, {Inspired by Finance: The Musiela Festschrift}, Y.
Kabanov et al. eds., Springer, Berlin, (2013), pp. 475--505.


\bibitem{OS2011} {{\sc B. \O ksendal and A. Sulem}, {\em Portfolio optimization under model uncertainty and BSDE games}, {Quantitiative Finance}, 11(11) (2011), pp. 1665--1674.}




\bibitem{S2007} {\sc A. Schied}, {\em Optimal investments for risk- and ambiguity-averse preferences: a duality approach}, {Finance Stoch.},
11 (2007), pp. 107--129.



\bibitem{SSZ} {\sc M. Shkolnikov, R. Sircar, and T. Zariphopoulou},
{\em Asymptotic analysis of forward performance processes in
incomplete markets and their ill-posed HJB\ equations}, {SIAM J.
Financ Math.}, 7(1) (2016), pp. 588--618.

\bibitem{SZ} {{\sc S. S. Strub and X. Y. Zhou} {\em Evolution of the Arrow-Pratt measure of risk-tolerance for predictable forward utility
processes}, {Finance Stoch.}, to appear.}

\bibitem{TZ2002} {\sc D. Talay and Z. Zheng}, {\em Worst case model risk management}, {Finance Stochast.},
6 (2002), pp. 517--537.


\bibitem{YLZ} {\sc Z. Yang, G. Liang, and C. Zhou},
{\em Constrained portfolio-consumption strategies with uncertain parameters and borrowing costs}, {Math. Financ. Econ.}, 13(3) (2018), pp. 393--427.




\bibitem{YYY2016} {{\sc S. C. P. Yam, H. Yang, and F. L. Yuen}, {\em Optimal asset allocation: Risk and information uncertainty}, {European Journal of Operational Research}, 251(2) (2016), pp. 554--561.}

\bibitem{Gordan} {\sc G. Zitkovic}, {\em A dual characterization of
self-generation and exponential forward performances}, {Ann. Appl.
Probab.}, 19(6) (2009), pp. 2176--2210.

\end{thebibliography}
\end{document}